\documentclass{icga}

\usepackage[english]{babel}
\usepackage{amssymb}
\usepackage{graphicx}
\usepackage{color}
\usepackage[usenames,dvipsnames]{xcolor}
\usepackage{gastex} 
\usepackage{subfigure}
\usepackage{url}

\usepackage{tikz}

\usetikzlibrary{decorations.markings}
\tikzset{middlearrow/.style={
        decoration={markings,
            mark= at position 0.6 with {\arrow{#1}} ,
        },
        postaction={decorate}
    }
}
\tikzset{middledoublearrow/.style={
        decoration={markings,
            mark= at position 0.80 with {\arrow{#1}} ,
        },
        postaction={decorate}
    }
}

\tikzstyle{blue_edge}=[middledoublearrow={>>}, line width=5pt, line join=round, line cap=round]
\tikzstyle{red_edge}=[middlearrow={>}, line width=1.5pt, line join=round, line cap=round]
\tikzstyle{blue_edge_rev}=[middledoublearrow={<<}, line width=5pt, line join=round, line cap=round]
\tikzstyle{red_edge_rev}=[middlearrow={<}, line width=1.5pt, line join=round, line cap=round]
\tikzstyle{blue_unedge}=[line width=5pt, line join=round, line cap=round]
\tikzstyle{red_unedge}=[line width=1.5pt, line join=round, line cap=round]

\definecolor{DarkGray}{rgb}{0.2, 0.2, 0.2}

\newtheorem{theorem}{Theorem}[section]

\newenvironment{proof}[1][Proof]{\begin{trivlist}
\item[\hskip \labelsep {\bfseries #1}]}{\end{trivlist}}

\title{ACYCLIC CONSTRAINT LOGIC AND GAMES}
\runningtitle{Acyclic Constraint Logic and Games}
\author{Hendrik Jan Hoogeboom\thanks{h.j.hoogeboom@liacs.leidenuniv.nl} \and Walter A. Kosters\thanks{w.a.kosters@liacs.leidenuniv.nl} \and Jan N. van Rijn\thanks{j.n.van.rijn@liacs.leidenuniv.nl} \and Jonathan K. Vis\thanks{j.k.vis@liacs.leidenuniv.nl}}
\affiliation{Leiden Institute of Advanced Computer Science, Universiteit Leiden, The Netherlands} 

\issue{March 2014}
\begin{document}
\maketitle
\setcounter{page}{3}

\begin{abstract}
Non-deterministic Constraint Logic is a family of graph games introduced by Demaine and Hearn that facilitates the construction of complexity proofs. It is convenient for the analysis of games, providing a uniform view.
We focus on the acyclic version, apply this to Klondike, Mahjong Solitaire and Nonogram (that requires planarity), and discuss
the more complicated game of Dou Shou Qi. While for the first three
games we reobtain known characterizations in a simple and uniform
manner, the result for Dou Shou Qi is new.
\end{abstract}

\section{Introduction}
\label{sec:NCL}

Besides actually playing games, it is of great interest to know how
hard these games are in the sense of computational complexity,
see~\citeaby{Kendall}. The games are usually generalized
to allow for parameters that control board size, number of cards, etc.

In order to study the structural complexity of games, \citeaby{Hearn2006} and \citeaby{Hearn2009} advocate the use of the constraint logic framework.
It consists of a collection of abstract graph games. The games are played on a so-called \emph{constraint graph}. A constraint graph is a weighted directed graph, where each edge has a weight in $\{1,2\}$. The \emph{inflow} of a vertex is defined to be the sum of all weights of the edges that are directed inward. A configuration (i.e., direction of the edges) of a constraint graph is legal if and only if for each vertex it holds that the inflow is at least its minimum inflow, usually 2. A move of a player is typically the reversal of one of the edges; players are only allowed to do moves that result in a legal configuration. 

A notable feature of the constraint logic framework is the fact that constraint graphs can be reduced to equivalent planar versions. Many real-life games are played on a 2-dimensional board. In previous game complexity results (e.g.,~\citeaby{Culberson1999,Flake2002}) crossover gadgets are necessary to overcome the limitations of such a 2-dimensional game board. Crossover gadgets are in general complex and hard to construct. The generic crossover gadget for constraint logic, as presented in Figure~\ref{fig:crossover} below, removes the need to devise a specific crossover gadget for every single game. 

Various games based on constraint graphs are defined in~\citeaby{Hearn2009}. These are categorized based on the number of players and whether there is a bound on the number of moves. We will describe two of those: \emph{Bounded Non-deterministic Constraint Logic} and \emph{Bounded Two-Player Constraint Logic}. In particular we will also pay attention to acyclic versions and planarity issues.

The remainder of this paper is organized as follows.
First we explain the different types of graph games. Next we apply these to the games Klondike, Mahjong Solitaire, Nonogram and Dou Shou Qi.

\subsection{Bounded Non-deterministic Constraint Logic}
\label{sec:BoundedNCL}
Bounded Non-deterministic Constraint Logic (Bounded NCL) is a
one-player game (i.e., a puzzle), played on a constraint graph. A move
is defined to be the reversal of one of the edges, resulting in a
legal configuration, i.e., meeting the inflow condition of the vertices.
Each edge may be reversed at most once. This puts an upper bound on the number of moves in this game, i.e., the number of edges in the graph. One of the edges is defined to be the \emph{target edge}; the player wins if and only if (s)he is able to reverse the target edge. 

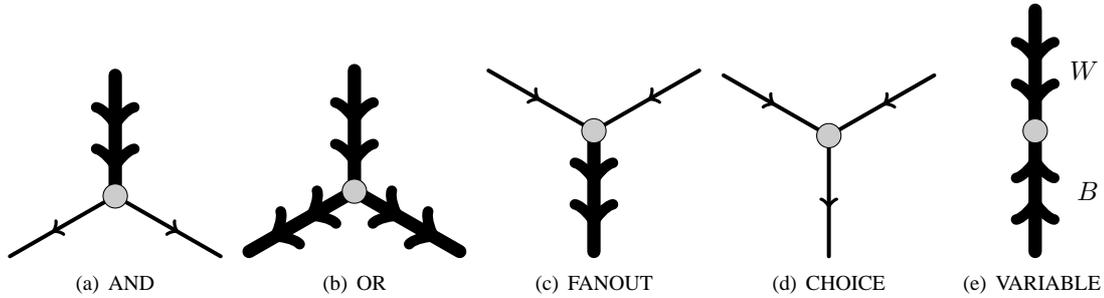
\begin{figure}[!ht]
 \centering
  \subfigure[AND]{
    \begin{tikzpicture}[y=-1cm,scale=0.8]
    \draw [blue_edge] (0, -2) -- (0, 0);
    \draw [red_edge] (0, 0) -- (210:2cm);
    \draw [red_edge] (0, 0) -- (330:2cm);
    \fill [black!20] (0, 0) circle (.2cm);
    \draw (0, 0) circle (.2cm);
    \end{tikzpicture}
    \label{fig:NCLAND}}
  \subfigure[OR]{
    \begin{tikzpicture}[y=-1cm,scale=0.8]
    \draw [blue_edge] (0, -2) -- (0, 0);
    \draw [blue_edge] (0, 0) -- (210:2cm);
    \draw [blue_edge] (0, 0) -- (330:2cm);
    \fill [black!20] (0, 0) circle (.2cm);
    \draw (0, 0) circle (.2cm);
    \end{tikzpicture}
    \label{fig:NCLOR}}
  \subfigure[FANOUT]{
    \begin{tikzpicture}[y=-1cm,scale=0.8]
    \draw [blue_edge] (0, 0) -- (0, 2);
    \draw [red_edge_rev] (0, 0) -- (30:2cm);
    \draw [red_edge_rev] (0, 0) -- (150:2cm);
    \fill [black!20] (0, 0) circle (.2cm);
    \draw (0, 0) circle (.2cm);
    \end{tikzpicture}
    \label{fig:NCLFANOUT}}
  \subfigure[CHOICE]{
    \begin{tikzpicture}[y=-1cm,scale=0.8]
    \draw [red_edge] (0, 0) -- (0, 2);
    \draw [red_edge_rev] (0, 0) -- (30:2cm);
    \draw [red_edge_rev] (0, 0) -- (150:2cm);
    \fill [black!20] (0, 0) circle (.2cm);
    \draw (0, 0) circle (.2cm);
    \end{tikzpicture}
    \label{fig:NCLCHOICE}}
  \subfigure[VARIABLE]{
    \begin{tikzpicture}[y=-1cm,scale=0.8]
    \draw [blue_edge] (0, -2) -- (0, 0);
    \draw [blue_edge] (0, 2) -- (0,0);
    \fill [black!20] (0, 0) circle (.2cm);
    \draw (0, 0) circle (.2cm);
    \node[left] at (1.2,-1) {$W$};
    \node[left] at (1.2,1) {$B$};
    \node[left] at (-1,-1) {\ };
    \end{tikzpicture}
    \label{fig:2CLVARIABLE}}
  \caption[Basic vertices]{Basic vertices, based on Figure~5.2 and Figure~6.2 from~\citeaby{Hearn2006}. Edges with a weight of $2$ use thick lines and have double arrows; edges with a weight of $1$ use thin lines and have a single arrow. Usually these edges are referred to as ``blue'' and ``red'', respectively.}
  \label{fig:thegadgets}
\end{figure}

Theorem~5.1 and Theorem~5.2 from \citeaby{Hearn2009} show (using a reduction from the Boolean satisfiability problem) that the game is NP-complete, even when the initial constraint graph only consists of AND, OR, FANOUT and CHOICE vertices as shown in Figure~\ref{fig:thegadgets}.

\subsection{Bounded Two-Player Constraint Logic}
\label{sec:Bounded2CL}
Bounded Two-Player Constraint Logic (Bounded 2CL) is a two-player perfect-information game played on a constraint graph, and a partitioning of the edges in disjoint sets $B$ and $W$. The players alternate turns. The white player reverses edges in $W$; the black player reverses edges in $B$. For both players it holds that their move has to result in a legal configuration. Each edge may only be reversed once, which (as in Bounded NCL) puts an upper bound on the number of moves in the game. One of the edges in $W$ is defined to be the target edge. The white player wins if (s)he is able to reverse this edge; if a player is unable to move, (s)he loses the game. 

Theorem 6.2 from \citeaby{Hearn2009} shows that the game is PSPACE-complete, even when the constraint graph only consists of the five vertices as shown in Figure~\ref{fig:thegadgets}, where the edges from AND, OR, FANOUT and CHOICE vertices are all in the set $W$. 
Note that the black player can only play bottom edges in VARIABLE
gadgets. In order to avoid clear loss for black an ample amount of
additional black edges is supplied.

\subsection{Acyclic graphs and crossover gadgets}
\label{sec:crossover}
In order to planarize constraint graphs, the construction shown in
Figure~\ref{fig:crossoverFull} can be used. A pair of crossing edges
can be replaced by this gadget. To obtain basic vertices as in
Figure~\ref{fig:thegadgets}, each vertex with four red edges can be
replaced by the so-called half-crossover gadget, which is shown in
Figure~\ref{fig:crossoverHalf}. Additionally, we need to perform
red-blue edge conversions, see~\citeaby{Hearn2006}.

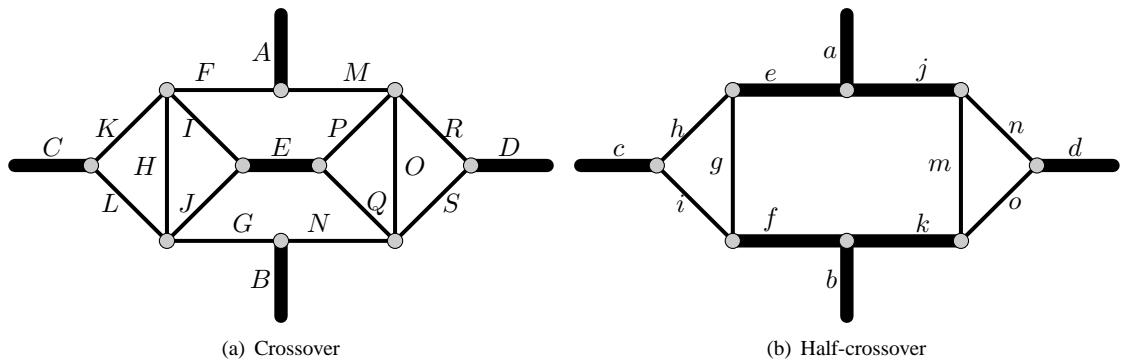
\begin{figure}[!ht]
 \centering
  \subfigure[Crossover]{
    \begin{tikzpicture}[y=-1cm]

    \node [above] at (0.5, 2) {$C$};
    \draw [blue_unedge] (0, 2) -- (1, 2);
    \node [above] at (3.5, 2) {$E$};
    \draw [blue_unedge] (3, 2) -- (4, 2);
    \node [left] at (3.5, 0.5) {$A$};
    \draw [blue_unedge] (3.5, 1) -- (3.5, 0);
    \node [left] at (3.5, 3.5) {$B$};
    \draw [blue_unedge] (3.5, 3) -- (3.5, 4);
    \node [above] at (6.5, 2) {$D$};
    \draw [blue_unedge] (6, 2) -- (7, 2);

    \node [left] at (1.5, 1.5) {$K$};
    \draw [red_unedge] (1, 2) -- (2, 1);
    \node [left] at (1.5, 2.5) {$L$};
    \draw [red_unedge] (1, 2) -- (2, 3);
    \node [left] at (2, 2) {$H$};
    \draw [red_unedge] (2, 1) -- (2, 3);
    \node [left] at (2.5, 1.5) {$I$};
    \draw [red_unedge] (2, 1) -- (3, 2);
    \node [left] at (2.5, 2.5) {$J$};
    \draw [red_unedge] (2, 3) -- (3, 2);
    \node [above] at (2.5, 1) {$F$};
    \draw [red_unedge] (2, 1) -- (3.5, 1);
    \node [above] at (3, 3) {$G$};
    \draw [red_unedge] (2, 3) -- (3.5, 3);
    \node [above] at (4.5, 1) {$M$};
    \draw [red_unedge] (3.5, 1) -- (5, 1);
    \node [above] at (4, 3) {$N$};
    \draw [red_unedge] (3.5, 3) -- (5, 3);
    \node [left] at (4.5, 1.5) {$P$};
    \draw [red_unedge] (4, 2) -- (5, 1);
    \node [right] at (4.5, 2.5) {$Q$};
    \draw [red_unedge] (4, 2) -- (5, 3);
    \node [right] at (5, 2) {$O$};
    \draw [red_unedge] (5, 1) -- (5, 3);
    \node [right] at (5.5, 1.5) {$R$};
    \draw [red_unedge] (5, 1) -- (6, 2);
    \node [right] at (5.5, 2.5) {$S$};
    \draw [red_unedge] (5, 3) -- (6, 2);

    \draw [fill=black!20] (1, 2) circle (.1cm);
    \draw [fill=black!20] (2, 1) circle (.1cm);
    \draw [fill=black!20] (2, 3) circle (.1cm);
    \draw [fill=black!20] (3, 2) circle (.1cm);
    \draw [fill=black!20] (3.5, 1) circle (.1cm);
    \draw [fill=black!20] (3.5, 3) circle (.1cm);
    \draw [fill=black!20] (4, 2) circle (.1cm);
    \draw [fill=black!20] (5, 1) circle (.1cm);
    \draw [fill=black!20] (5, 3) circle (.1cm);
    \draw [fill=black!20] (6, 2) circle (.1cm);

    \end{tikzpicture}

    \label{fig:crossoverFull}}
  \subfigure[Half-crossover]{
    \begin{tikzpicture}[y=-1cm]

    \node [above] at (0.5, 2) {$c$};
    \draw [blue_unedge] (0, 2) -- (1, 2);
    \node [above] at (2.5, 1) {$e$};
    \draw [blue_unedge] (2, 1) -- (3.5, 1);
    \node [above] at (4.5, 1) {$j$};
    \draw [blue_unedge] (3.5, 1) -- (5, 1);
    \node [left] at (3.5, 0.5) {$a$};
    \draw [blue_unedge] (3.5, 1) -- (3.5, 0);
    \node [above] at (2.5, 3) {$f$};
    \draw [blue_unedge] (2, 3) -- (3.5, 3);
    \node [above] at (4.5, 3) {$k$};
    \draw [blue_unedge] (3.5, 3) -- (5, 3);
    \node [left] at (3.5, 3.5) {$b$};
    \draw [blue_unedge] (3.5, 3) -- (3.5, 4);
    \node [above] at (6.5, 2) {$d$};
    \draw [blue_unedge] (6, 2) -- (7, 2);

    \node [left] at (1.5, 1.5) {$h$};
    \draw [red_unedge] (1, 2) -- (2, 1);
    \node [left] at (1.5, 2.5) {$i$};
    \draw [red_unedge] (1, 2) -- (2, 3);
    \node [left] at (2, 2) {$g$};
    \draw [red_unedge] (2, 1) -- (2, 3);
    \node [right] at (5.5, 1.5) {$n$};
    \draw [red_unedge] (6, 2) -- (5, 1);
    \node [right] at (5.5, 2.5) {$o$};
    \draw [red_unedge] (6, 2) -- (5, 3);
    \node [left] at (5, 2) {$m$};
    \draw [red_unedge] (5, 1) -- (5, 3);

    \draw [fill=black!20] (1, 2) circle (.1cm);
    \draw [fill=black!20] (2, 1) circle (.1cm);
    \draw [fill=black!20] (2, 3) circle (.1cm);
    \draw [fill=black!20] (3.5, 1) circle (.1cm);
    \draw [fill=black!20] (3.5, 3) circle (.1cm);
    \draw [fill=black!20] (5, 1) circle (.1cm);
    \draw [fill=black!20] (5, 3) circle (.1cm);
    \draw [fill=black!20] (6, 2) circle (.1cm);

    \end{tikzpicture}
    \label{fig:crossoverHalf}}
  \caption[Planar crossover gadgets]{Planar crossover gadgets, as presented in~\citeaby{Hearn2009}.}
  \label{fig:crossover}
\end{figure}

Edges $A$ and $B$ are called the \emph{vertical external edges}. Edges $C$ and $D$ are called the \emph{horizontal external edges}. 
It can be verified that each vertical external edge can point outward if and only if the other vertical external edge points inward. A similar property holds for the horizontal external edges. The action of reversing both vertical external edges is called \emph{vertical propagation}; the action of reversing both horizontal external edges is called \emph{horizontal propagation}. For example, when the edges $A$ and $B$ are pointing up, and the edges $C$ and $D$ are pointing left, the direction of all other edges follows from the inflow constraints. A sequence of, e.g., $(A, F, H, G, M, O, N, B)$ would then perform a vertical propagation; a sequence of, e.g., $(C, K, I, L, J, E, P, R, Q, S, D)$ would perform a horizontal propagation. 

Although this gadget indeed simulates all the behavior of the games introduced in, 
e.g.,~\citeaby{Hearn2005} (where the same edge can be reversed multiple times) 
this is not the case for the constraint graphs used in Bounded NCL and Bounded 2CL. 
After performing a vertical propagation, due to the restriction that each edge may 
be reversed at most once, the gadget is in such a state that it is impossible to 
perform horizontal propagation, and vice versa. 
Although the construction in Figure~\ref{fig:crossoverFull} is valid, 
after integrating the construction of Figure~\ref{fig:crossoverHalf}
for all the vertices with four red edges 
(in order to restrict ourselves to the gadgets in Figure~\ref{fig:thegadgets}) 
it becomes clear that this is no longer the case. 
After, e.g., edges $F$ and $H$ (corresponding to $a$ and $b$, respectively) 
are reversed, it can be verified that due to the internal state of the half-crossover, 
edges $K$ and $I$ (corresponding to $c$ and $d$, respectively) cannot be reversed 
anymore. The only way to perform both a horizontal and vertical propagation over the 
same crossover gadget is when both a horizontal external edge and vertical external 
edge can be reversed inward at the same moment: a typical example of a \emph{race condition}.
An extensive analysis of the properties of the crossover and half-crossover gadgets as well as red-blue edge converters can be found in~\citeaby{Hearn2006} and~\citeaby{Rijn2012}.





It is clear that not all constraint graphs can be reduced to a planar equivalent using solely the gadgets presented in Figure~\ref{fig:crossover}. In some configurations, in particular in (initially) cyclic graphs, it is impossible to obey the additional constraint imposed by the crossover gadget that the propagation of both directions has to happen at the same moment. However, for acyclic graphs, this is never a problem. Edges can be topologically sorted, and reversed in this order. The complexity proofs of both Bounded NCL and Bounded 2CL (see~\citeaby{Hearn2009}) use graphs corresponding to logical formulas, which indeed require only acyclic graphs; hence Planar Bounded NCL is NP-complete and Planar Bounded 2CL is PSPACE-complete.
Note that all graphs under consideration are acyclic in their initial configuration.

Theorem 5.4 from~\citeaby{Hearn2009} states that the related problem Constraint Graph Satisfiability is also NP-complete:
does a given planar constraint graph, using only (initially undirected) AND and OR vertices, have a legal configuration? Note
that this strictly speaking is not a game in the above sense: we only
ask for a legal ``final'' configuration, not the sequence of moves that can
be used to obtain it.

\section{Klondike}
\label{sec:klondike}
\emph{Klondike}, also known as Patience or Solitaire, is a well-known card game, popularized by Microsoft Windows. The normal version of the game is played with a standard French card deck, without jokers. \citeaby{Yan2004} 
have given a formal definition of the game and provided an algorithm that plays Klondike games with a high success rate;
in their version of the game, often referred to as thoughtful Solitaire, the identity of all cards is known from the beginning.
Several other approaches have been proposed to deal with Klondike, see, e.g., \citeaby{Bjarnason1} and \citeaby{Bjarnason2}.

\citeaby{Longpre2009} have shown, amongst other complexity results, that Klondike is NP-complete even when played with two red suits and one black suit: red diamonds ($\diamondsuit$), red hearts ($\heartsuit$) and black spades ($\spadesuit$). We will give a formal definition of the necessary subset of Klondike and confirm the NP-completeness of Klondike by a reduction from Acyclic Bounded NCL, using an argument improved upon the one from~\citeaby{Rijn2012}.

\subsection{Definition}
Generalized Klondike is played with a card deck containing $m$ suits, each suit containing $n$ cards ranked from $1$ to $n$. A card with rank $1$ is also referred to as an \emph{Ace}; a card with rank $n$ is also referred to as a \emph{King}. The functions $\mathit{rank}(c)$ and $\mathit{suit}(c)$ return the rank and the suit of card $c$, respectively. Each suit is colored either red or black. The function $\mathit{color}(s)$ returns the color of suit $s$.

A Klondike game consists of $m$ \emph{suit stacks}, one or more \emph{build stacks}, a \emph{pile stack} and a \emph{talon}. In the sequel we do not need pile stack and talon, so these will be omitted from the description. A stack is defined to be an ordered list of cards. A \emph{configuration} describes for each card in which stack it is and on which position. For every card in a build stack it also describes whether the card is \emph{face-up} or \emph{face-down}. The subset of cards that are face-up on a certain build stack constitute a \emph{card block}, and will always consist of topmost cards. In an initial configuration all cards are face-down in the build stacks (that can be of different lengths), and the suit stacks are empty. 

We will define the notion of \emph{acceptance}, which determines which moves the player can make. Each suit stack that is empty can only accept an Ace. Every suit stack that is not empty, containing card $t$ on top, accepts card $c$ if and only if $\mathit{suit}(c) = \mathit{suit}(t)$, and $\mathit{rank}(c) = \mathit{rank}(t) + 1$. Therefore, suit stacks accept cards of the same suit in ascending order. Each card block that is not empty, containing card $t$ on top, accepts card $c$ if and only if $\mathit{color}(\mathit{suit}(t))\neq\mathit{color}(\mathit{suit}(c))$ and $\mathit{rank}(c)=\mathit{rank}(t) - 1$. Therefore, build stacks accept cards in descending order, of alternating colors. We will not employ the usual property that an empty
build stack (only) accepts a King.

On each turn, the player can play cards in the following manner:
\begin{enumerate}
 \item If all cards on a build stack are face-down, the card on top can be turned face-up, thereby creating a singleton card block.
 \item A whole card block $p$ can be moved to the top of another card block $q$, provided that $q$ accepts the card at the bottom of $p$. (In some versions of the game a partial card block can also be moved in this manner.)
\item The top card $c$ of a card block can be moved to a suit stack, provided that the suit stack accepts $c$.
\end{enumerate}
The goal is to move all cards to the suit stacks, and when this is achieved the player has won.

\subsection{NP-completeness}
In order to prove NP-completeness, we will show that every Acyclic Bounded NCL graph can be transformed to a Klondike configuration, in such a way that the Klondike game can be won if and only if the target edge of the Acyclic Bounded NCL graph can be flipped. 
So we study the corresponding decision problem \textsc{Klondike}: given a Klondike configuration, can the player win?

\begin{figure}[!ht]
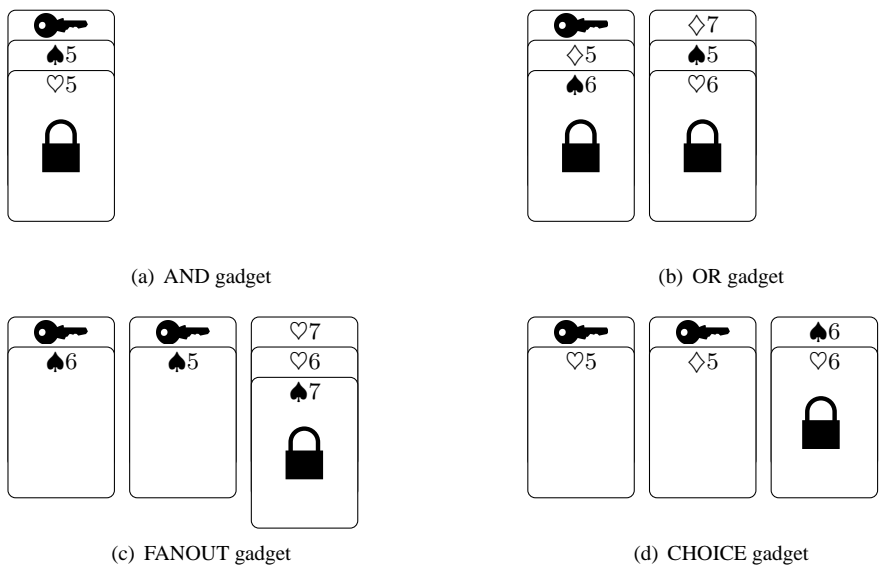

 \centering
  \subfigure[AND gadget\ \ \ \ \ \ \ \ \ \ \ \ \ \ \ \ \ ]{
    \input{images/gadgets_klondike//and.tex}
    \label{fig:patienceAnd}}
    \hspace{2mm}
  \subfigure[OR gadget\ \ \ \ \ \ \ \ \ \ \ \ \ \ \ \ \ ]{
    \input{images/gadgets_klondike//or.tex}
    \label{fig:patienceOr}}

  \subfigure[FANOUT gadget\ \ \ \ \ \ \ \ \ \ \ \ \ \ \ \ \ ]{
    \input{images/gadgets_klondike/fanout.tex}
    \label{fig:patienceFanout}}
    \hspace{2mm}
  \subfigure[CHOICE gadget\ \ \ \ \ \ \ \ \ \ \ \ \ \ \ \ \ ]{
    \input{images/gadgets_klondike//choice.tex}
    \label{fig:patienceChoice}}
  \caption[Klondike gadgets]{Klondike gadgets. Note that $\spadesuit5$ in the AND gadget is also a lock card.}
  \label{fig:patienceGadgets}
\end{figure}

We will use the four gadgets from Figure~\ref{fig:patienceGadgets}. The gadgets consist of one, two or three build stacks, with all cards initially face-down. Each gadget gets a range of unique ranks assigned to it; for simplicity, in the figure we use the range 5--8 for all gadgets.
A \emph{lock card} represents the tail of an edge adjacent to the corresponding NCL vertex; a lock icon is displayed on these cards. A \emph{key card} represents the head of an edge adjacent to the corresponding NCL vertex; a key icon is displayed on these cards. The rank of each key card is within the range of another gadget. The suit and rank of a key card is chosen such that once turned face-up, it accepts a lock card of the gadget from which the corresponding NCL edge is pointing (locks must be moved to their keys). The gadgets are constructed in such a way that the key card can be turned face-up if and only if the corresponding NCL edge can be flipped. Note that in the AND gadget $\spadesuit5$ is also a lock card (without having a lock image).

For each lock card $\ell$ it is easy to see which card should be turned face-up in order to move it. If $\mathit{color}(\mathit{rank}(\ell))$ is red, this card is $\spadesuit(\mathit{rank}(\ell) + 1)$, and otherwise it is $\heartsuit(\mathit{rank}(\ell) + 1)$ (the card $\diamondsuit(\mathit{rank}(\ell) + 1)$ will be only made available during the end play, or in the case of the OR gadget is positioned deeper within the gadget). These cards serve as key cards in other gadgets.

Now we note that the four gadgets indeed act as intended.
For instance, consider the OR gadget. In order to turn the key card face-up,
either the lock card $\spadesuit6$ (followed by $\diamondsuit5$) must be moved to
its corresponding key card $\heartsuit7$, or the lock card $\heartsuit6$ (followed by $\spadesuit5$) must be moved to
its corresponding key card $\spadesuit7$ after which  $\spadesuit6$ and $\diamondsuit5$ can
be moved to $\diamondsuit7$. If both key cards are available, both sequences can be played.
The AND gadget has a fixed order to free the key card, which is sufficient for our purpose.

Now we have:

\begin{theorem}
\textsc{Klondike} is NP-complete.
\end{theorem}

\begin{proof}
Reduction from Acyclic Bounded NCL. Given a constraint graph made of AND, OR, FANOUT and CHOICE vertices, we construct a corresponding Klondike configuration using the gadgets shown in Figure~\ref{fig:patienceGadgets}.
Note that planarity is not an issue both here and for Mahjong.

We need a way to ensure that the player can move all cards to the suit stacks if and only if the key card corresponding to the target edge can be turned face-up. To this end,
all cards not used in the gadgets are positioned in one big build stack (ordered by rank and within each rank in $\diamondsuit \heartsuit \spadesuit$ order, with the three Aces at the top and ending with the three Kings at the bottom), protected by a lock card
representing the target edge. Once this card is moved, all these other cards become available
and allow all cards from all gadgets to be moved to the suit stacks. 

Now the fact that the original NCL graph is acyclic is used. Indeed, the gadgets can be numbered,
using a topological sort of the corresponding nodes, and we can take care that for every
gadget the key cards used have higher rank than the cards in the gadgets.
This ensures a proper order for this part of the process.
In fact, even (partially) unplayed gadgets can be ``discarded'', using
the inductive assumption that all cards of lower ranks have already been moved to the suit stacks.
Note that,
if for the OR gadget $\heartsuit5$ were used instead of $\diamondsuit5$, this property would not hold.

For creating the Klondike configuration, the number of cards and stacks we need are both bounded by a linear function of the number of vertices in the corresponding NCL graph. In a winning sequence there are exactly $mn$ moves of type 1, and $mn$ moves of type 3.
As for the type 2 moves, there are at most $mn$ of them: every card block is moved once (when focussing on its bottom card), maybe to a suit stack. Therefore Klondike is in NP, since any potential solution can be verified in polynomial time. \hfill$\Box$

\end{proof}

\section{Mahjong Solitaire}\label{sec:mahjong}

\emph{Mahjong Solitaire}, also known as Shanghai Solitaire, is a one
player puzzle game mainly played on the computer in which the player
is presented with a randomly arranged stack of tiles. The goal is to remove
all tiles in matching pairs of two. \citeaby{Condon1997} have given a
formal definition of this game, and have shown that a version of this
game with imperfect information is PSPACE-complete.
\citeaby{Eppstein2012} has stated a proof that a version of this game
with perfect information is NP-complete. In the paper
by~\citeaby{Bondt2012} Mahjong is proven to be NP-complete by a
reduction from 3-SAT. We will give a formal definition of this game
and validate the latter result by a reduction from Acyclic Bounded
NCL.

\subsection{Definition}

The game uses Mahjong tiles, that are divided into $m$ disjoint
\emph{tile sets} ${\cal{T}}_p$ of $|{\cal{T}}_p| = s_p$ matching tiles,
where $s_p$ is an even number ($p = 1,2,\ldots,m$). We define the set
of all tiles to be ${\cal{T}} = \bigcup_{p} {\cal{T}}_p$.
Two tiles $a$ and $b$ \emph{match}, if and only if for some $p$ it holds
that $a,b \in {\cal{T}}_p$. Below we say that elements of the same
tile set have the same color. We generalize the standard game simply by
assuming that there is an arbitrarily large, finite number of tiles.
A \emph{configuration} $C$ is a set of positions $(i, j, k)$, where
each of $i,j,k$ is a non-negative integer, satisfying the following
constraints:
\begin{enumerate}
\item If $(i,j,k) \in C$ and $(i,j',k) \in C$ where $j < j'$, then for every $j''$ in the range 
[$j,j'$], $(i,j'',k) \in C$;
\item If $(i,j,k) \in C$ where $k > 0$ then $(i,j,k-1) \in C$. 
\end{enumerate}

This captures the fact that tiles are arranged in three
dimensions. Tiles can be stacked on top of each other; all tiles with
common $k$ are at the same height. All tiles with a common $i$~index,
form a \emph{cross section}. Tiles at the same height, with common
$i$~index, form a \emph{row}. The first condition ensures that there
cannot be gaps in a row; the second, that a tile at height $k > 0$
must have a tile underneath it (in fact, at position $(i,j,k-1$)).

With respect to a given configuration, a position $(i,j,k)$ is
\emph{hidden} if in the configuration also a position $(i,j,k+1)$
exists; the other positions are called \emph{visible}. An
\emph{arrangement} consists of a set of tiles ${\cal{T}}$, a
configuration $C$ of size~$|{\cal{T}}|$, and a bijective function $f$ from
the positions of $C$ to ${{\cal{T}}}$. Here
$f(i,j,k)$ denotes the tile at position $(i,j,k)$. If the function $f$ maps
position $(i,j,k)$ to tile $t$ we say $\mathit{pos}(t) = (i,j,k)$. The
elements of ${{\cal{T}}}$ will be mapped to the elements of $C$ in
such a way, that every combination is possible. With
respect to a given arrangement, we say a position $(i,j,k)$ is
\emph{available} if it is not hidden, and either position
$(i,j-1,k) \notin C$ or position $(i,j+1,k) \notin C$ or both, i.e.,
we can only take tiles that are at one of the ends of a row, and that have no tiles on top of it. An
arrangement is called empty if ${\cal{T}}$ is empty.
In order to avoid misunderstandings, all tiles can be seen from the beginning:
the player has perfect information.

A legal move consists of the removal of two matching tiles $a,b$ that
are both available. Formally, ${\cal{T}'} = {\cal{T}} - \{a,b\}$ and
$C' = C - \{\mathit{pos}(a),\mathit{pos}(b)\}$. The game is won if a
series of moves results in the empty arrangement.

\subsection{NP-completeness}
In order to prove NP-completeness, we will show that every Acyclic
Bounded NCL graph can be transformed to a Mahjong configuration, in
such a way that the Mahjong game can be won if and only if the target
edge of the Acyclic Bounded NCL graph can be flipped.
We study the decision problem \textsc{Mahjong Solitaire}: given a
Mahjong configuration, can the player win?

\begin{figure}[!ht]
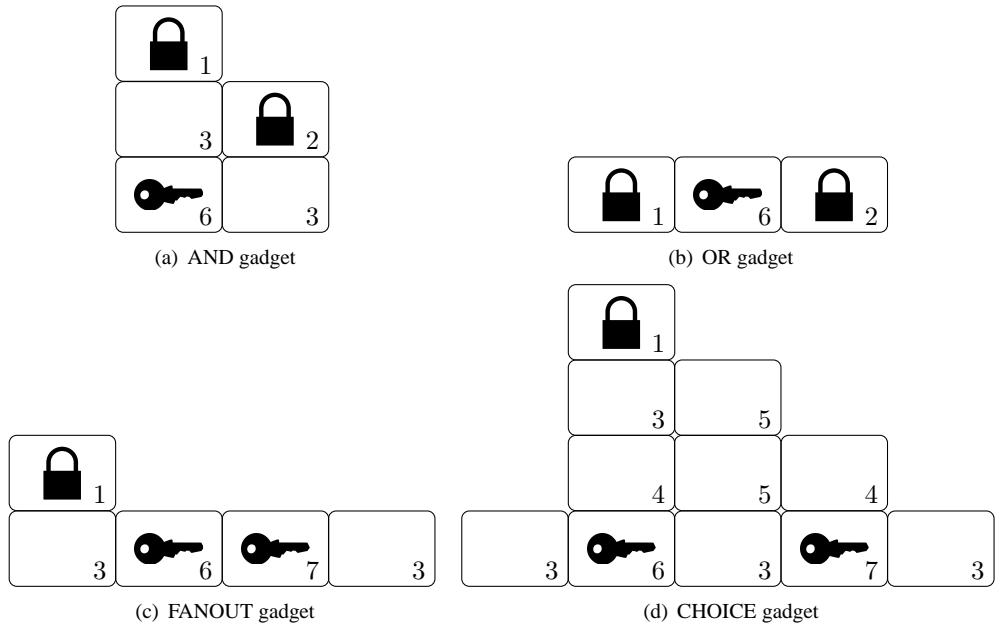

\begin{center}
\subfigure[AND gadget]{
\input{images/gadgets_mahjong/and.tex}
\label{fig:mahjongAnd}}
\subfigure[OR gadget]{
\input{images/gadgets_mahjong/or.tex}
\label{fig:mahjongOr}}
\subfigure[FANOUT gadget]{
\input{images/gadgets_mahjong/fanout.tex}
\label{fig:mahjongFanout}}
\subfigure[CHOICE gadget]{
\input{images/gadgets_mahjong/choice.tex}
\label{fig:mahjongChoice}}
\caption{Mahjong gadgets as a cross section of a configuration.}
\label{fig:mahjongGadgets}
\end{center}
\end{figure}

We will use the gadgets in Figure~\ref{fig:mahjongGadgets}. The
gadgets consist of one cross section containing between two and five tile stacks, with different numbers
representing different tile sets. 
Every gadget has a unique range of
numbers.
Again, for simplicity, in the figure we use the range 1--7 for all gadgets.
A \emph{lock tile} represents the tail
of an edge adjacent to the corresponding NCL vertex; a lock
icon is displayed on these tiles. A \emph{key tile} represents the
head of an edge which is adjacent to the corresponding NCL
vertex; a key icon is displayed on these tiles. 
Corresponding key and lock tiles from different gadgets share the
same number, for obvious reasons.

When a key tile is available, it can be removed together with a lock
tile from one of the other gadgets. The target edge in the
corresponding NCL graph is always represented by a key tile.
When this target edge is available this will initiate the end game,
which is always winning for the player as we will show further on.

The four gadgets in Figure~\ref{fig:mahjongGadgets} have their intended
behavior. For instance, consider the CHOICE gadget. To free either one
of the key tiles the lock tile has to be removed. Now, the actual
choice has to be made: the newly freed ``3-tile'' must be used to
remove either the leftmost or rightmost ``3-tile''. After removing both
``4-tiles'' precisely one of the key tiles is available.
Note that the gadgets resemble those for Klondike; in fact, the AND gadget can also be modeled to look
even more like its counterpart.

Now we have:

\begin{theorem}
\textsc{Mahjong Solitaire} is NP-complete.
\end{theorem}

\begin{proof}
Reduction from Acyclic Bounded NCL. Given a constraint graph made of
AND, OR, FANOUT and CHOICE vertices, we construct a corresponding
Mahjong configuration using the gadgets shown in
Figure~\ref{fig:mahjongGadgets}.

In order to have a way of clearing all remaining tiles after the target
tile is removed, we supply a victory gadget consisting of one linear
row of tiles with two ``5-tiles'' for every CHOICE gadget. 
The victory
gadget itself is protected in a similar way as the FANOUT gadget: a pair 
of matching tiles is placed at both sides, and a tile matching the target 
tile is placed on the left one of these. This ensures that none of the tiles
in the victory gadget can be used before the tiles representing the target 
edge are removed.
Again, we will use the fact that the original NCL graph is acyclic by
numbering the gadgets using the topological sort of the corresponding
vertices. This ordering defines the proper order for this process allowing
partially unplayed (CHOICE) gadgets to be removed by using the tiles from the
victory gadget, that are placed in this same order. Note that from this
acyclicity it follows that when the player is able to remove all
CHOICE gadgets, all other (partially) unplayed gadgets can also be removed. 

Both the number of tiles and the number of tile sets we need to use is
linearly bounded by the number of vertices used in the corresponding
Acyclic Bounded NCL graph. Therefore Mahjong Solitaire is in NP, since
any potential solution can be verified in polynomial time. \hfill$\Box$
\end{proof}

\section{Nonogram}
\label{sec:nonograms}
A \emph{Nonogram}, also referred to as a Japanese puzzle, is a logic 
puzzle which can be considered as an image reconstruction problem. 
The player is presented a rectangular grid; for each row and column 
a description consisting of one or more integers is provided, 
representing the numbers of consecutive cells that need to be black. 
If the player can color a subset of cells in such a way that it is 
consistent with the description of all rows and columns, (s)he has 
solved the puzzle and won the game. An example of a Nonogram and its 
solution is shown in Figure~\ref{fig:NonogramExample}.
\begin{figure}[!ht]
 \centering
  \subfigure[$6\times6$ Nonogram]{
    \begin{picture}(48, 48)
\linethickness{0.1mm}
\normalsize\node[Nw=6,Nh=6,linecolor=White](generated_hnode0)(3,33){2}
\node[Nw=6,Nh=6,linecolor=White](generated_hnode1)(9,33){1}
\node[Nw=6,Nh=6,linecolor=White](generated_hnode2)(3,27){1}
\node[Nw=6,Nh=6,linecolor=White](generated_hnode3)(9,27){3}
\node[Nw=6,Nh=6,linecolor=White](generated_hnode4)(3,21){1}
\node[Nw=6,Nh=6,linecolor=White](generated_hnode5)(9,21){2}
\node[Nw=6,Nh=6,linecolor=White](generated_hnode6)(9,15){3}
\node[Nw=6,Nh=6,linecolor=White](generated_hnode8)(9,9){4}
\node[Nw=6,Nh=6,linecolor=White](generated_hnode10)(9,3){1}
\node[Nw=6,Nh=6,linecolor=White](generated_vnode0)(15,39){1}
\node[Nw=6,Nh=6,linecolor=White](generated_vnode2)(21,39){5}
\node[Nw=6,Nh=6,linecolor=White](generated_vnode4)(27,39){2}
\node[Nw=6,Nh=6,linecolor=White](generated_vnode6)(33,39){5}
\node[Nw=6,Nh=6,linecolor=White](generated_vnode8)(39,45){2}
\node[Nw=6,Nh=6,linecolor=White](generated_vnode9)(39,39){1}
\node[Nw=6,Nh=6,linecolor=White](generated_vnode10)(45,39){2}
\multiput(12, 0)(6,0){7}{\color{DarkGray} \line(0,1){48}}
\multiput(0, 36)(0,-6){7}{\color{DarkGray} \line(1,0){48}}
\end{picture}
    \label{fig:NonogramExampleEmpty}}
  \subfigure[Solved Nonogram]{
    \begin{picture}(48, 48)
\linethickness{0.1mm}
\normalsize\node[Nw=6,Nh=6,linecolor=White](generated_hnode0)(3,33){2}
\node[Nw=6,Nh=6,linecolor=White](generated_hnode1)(9,33){1}
\node[Nw=6,Nh=6,linecolor=White](generated_hnode2)(3,27){1}
\node[Nw=6,Nh=6,linecolor=White](generated_hnode3)(9,27){3}
\node[Nw=6,Nh=6,linecolor=White](generated_hnode4)(3,21){1}
\node[Nw=6,Nh=6,linecolor=White](generated_hnode5)(9,21){2}
\node[Nw=6,Nh=6,linecolor=White](generated_hnode6)(9,15){3}
\node[Nw=6,Nh=6,linecolor=White](generated_hnode8)(9,9){4}
\node[Nw=6,Nh=6,linecolor=White](generated_hnode10)(9,3){1}
\node[Nw=6,Nh=6,linecolor=White](generated_vnode0)(15,39){1}
\node[Nw=6,Nh=6,linecolor=White](generated_vnode2)(21,39){5}
\node[Nw=6,Nh=6,linecolor=White](generated_vnode4)(27,39){2}
\node[Nw=6,Nh=6,linecolor=White](generated_vnode6)(33,39){5}
\node[Nw=6,Nh=6,linecolor=White](generated_vnode8)(39,45){2}
\node[Nw=6,Nh=6,linecolor=White](generated_vnode9)(39,39){1}
\node[Nw=6,Nh=6,linecolor=White](generated_vnode10)(45,39){2}
\node[Nw=6,Nh=6,fillcolor=Black,Nmr=0](generated_cell0)(15,33){}
\node[Nw=6,Nh=6,fillcolor=Black,Nmr=0](generated_cell1)(21,33){}
\node[Nw=6,Nh=6,fillcolor=Black,Nmr=0](generated_cell5)(45,33){}
\node[Nw=6,Nh=6,fillcolor=Black,Nmr=0](generated_cell7)(21,27){}
\node[Nw=6,Nh=6,fillcolor=Black,Nmr=0](generated_cell9)(33,27){}
\node[Nw=6,Nh=6,fillcolor=Black,Nmr=0](generated_cell10)(39,27){}
\node[Nw=6,Nh=6,fillcolor=Black,Nmr=0](generated_cell11)(45,27){}
\node[Nw=6,Nh=6,fillcolor=Black,Nmr=0](generated_cell13)(21,21){}
\node[Nw=6,Nh=6,fillcolor=Black,Nmr=0](generated_cell15)(33,21){}
\node[Nw=6,Nh=6,fillcolor=Black,Nmr=0](generated_cell16)(39,21){}
\node[Nw=6,Nh=6,fillcolor=Black,Nmr=0](generated_cell19)(21,15){}
\node[Nw=6,Nh=6,fillcolor=Black,Nmr=0](generated_cell20)(27,15){}
\node[Nw=6,Nh=6,fillcolor=Black,Nmr=0](generated_cell21)(33,15){}
\node[Nw=6,Nh=6,fillcolor=Black,Nmr=0](generated_cell25)(21,9){}
\node[Nw=6,Nh=6,fillcolor=Black,Nmr=0](generated_cell26)(27,9){}
\node[Nw=6,Nh=6,fillcolor=Black,Nmr=0](generated_cell27)(33,9){}
\node[Nw=6,Nh=6,fillcolor=Black,Nmr=0](generated_cell28)(39,9){}
\node[Nw=6,Nh=6,fillcolor=Black,Nmr=0](generated_cell33)(33,3){}
\multiput(12, 0)(6,0){7}{\color{DarkGray} \line(0,1){48}}
\multiput(0, 36)(0,-6){7}{\color{DarkGray} \line(1,0){48}}
\end{picture}
    \label{fig:NonogramExampleSolved}}
  \caption[Nonogram instance]{An example Nonogram~\subref{fig:NonogramExampleEmpty} and its unique solution~\subref{fig:NonogramExampleSolved}, taken from~\citeaby{Batenburg2012}.}
  \label{fig:NonogramExample}
\end{figure}
\citeaby{Batenburg2012} have given a formal definition of 
Nonograms and provided an algorithm for solving many Nonograms 
in polynomial time. In~\citeaby{Nagao1996} it is proven that the 
Another Solution Problem for Nonograms is NP-complete, and more in 
particular that the question whether a given puzzle has~a~solution~is 
NP-complete. We will also give a formal definition of Nonograms and 
show that the latter decision problem is~NP-complete, by reduction 
from Constraint Graph Satisfiability.

\subsection{Definition}
A Nonogram is a puzzle in which the player is presented an $m \times n$ grid of \emph{cells}, consisting of $m$ rows and $n$ columns. 
The state of a cell is either $\mathit{white}=0$ or $\mathit{black}=1$. Initially, all cells are $\mathit{white}$. 
A \emph{line} is defined to be either a row or a column.

For each line a description $d$ is provided, $d$ being an ordered series of integers $(d_1, d_2, \ldots, d_{k})$. 
The description is adhered to, if there are exactly $k$ \emph{black segments} in the line, where each successive black segment $s$ is of size $d_s$ ($s = 1,2,\ldots,k$). A black segment is defined to be a group of consecutive cells in the line, such that all cells within the interval are $\mathit{black}$, and both cells adjacent to the interval, if any, are $\mathit{white}$. 
Now the puzzle is solved if the player can make a subset of the cells black,
in such a way that all descriptions are adhered to. The corresponding decision problem
\textsc{Nonogram} asks if a given Nonogram can be solved.
For more information on Nonograms, the reader is referred to~\citeaby{Batenburg2012} and the references therein.

\subsection{NP-completeness}

We will show that solving Nonograms is NP-complete, by reduction from Constraint Graph Satisfiability (\citeaby{Hearn2009}), only using two initially undirected gadgets: AND and OR.
The global layout of the construction will be as in Figure~\ref{fig:NonogramGlobalLayout}. There will be several groups of $D$ adjacent columns (or rows) where the description consists of a single element, i.e., $m$ (or $n$), such that the pattern of Figure~\ref{fig:NonogramGlobalLayout} is maintained. We call these lines the \emph{separation lines}. Between each group of separation lines, there are $G$ other lines. In the case of Figure~\ref{fig:NonogramGlobalLayout}, $D = 5$ and $G = 7$. 
The descriptions and the width of the delimiters will not interfere with those of the
single elements in between.
As a result of this construction we can specify disjoint \emph{subnonograms} between the separation lines.

\begin{figure}[!ht]
 \centering
  \input{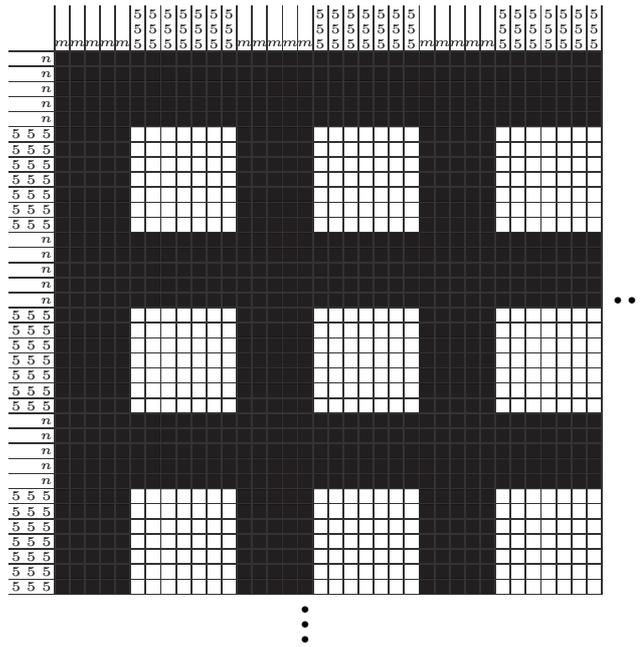}
  \caption[Global layout for Nonogram reduction]{Global layout.}
  \label{fig:NonogramGlobalLayout}
\end{figure}


\begin{figure}[!ht]
 \centering
  \subfigure[]{
    \begin{picture}(48, 20)
\linethickness{0.05mm}
\node[Nw=4,Nh=4,linecolor=White](generated_hnode0)(2,14){3}
\node[Nw=4,Nh=4,linecolor=White](generated_hnode1)(2,10){4}
\node[Nw=4,Nh=4,linecolor=White](generated_hnode2)(2,6){4}
\node[Nw=4,Nh=4,linecolor=White](generated_hnode3)(2,2){3}
\node[Nw=4,Nh=4,linecolor=White](generated_vnode0)(6,18){0}
\node[Nw=4,Nh=4,linecolor=White](generated_vnode1)(10,18){0}
\node[Nw=4,Nh=4,linecolor=White](generated_vnode2)(14,18){0}
\node[Nw=4,Nh=4,linecolor=White](generated_vnode3)(18,18){1}
\node[Nw=4,Nh=4,linecolor=White](generated_vnode4)(22,18){4}
\node[Nw=4,Nh=4,linecolor=White](generated_vnode5)(26,18){4}
\node[Nw=4,Nh=4,linecolor=White](generated_vnode6)(30,18){4}
\node[Nw=4,Nh=4,linecolor=White](generated_vnode7)(34,18){1}
\node[Nw=4,Nh=4,linecolor=White](generated_vnode8)(38,18){0}
\node[Nw=4,Nh=4,linecolor=White](generated_vnode9)(42,18){0}
\node[Nw=4,Nh=4,linecolor=White](generated_vnode10)(46,18){0}
\node[Nw=4,Nh=4,fillcolor=Black,Nmr=0](generated_cell4)(22,14){}
\node[Nw=4,Nh=4,fillcolor=Black,Nmr=0](generated_cell5)(26,14){}
\node[Nw=4,Nh=4,fillcolor=Black,Nmr=0](generated_cell6)(30,14){}
\node[Nw=4,Nh=4,fillcolor=Black,Nmr=0](generated_cell7)(18,10){}
\node[Nw=4,Nh=4,fillcolor=Black,Nmr=0](generated_cell8)(22,10){}
\node[Nw=4,Nh=4,fillcolor=Black,Nmr=0](generated_cell9)(26,10){}
\node[Nw=4,Nh=4,fillcolor=Black,Nmr=0](generated_cell10)(30,10){}
\node[Nw=4,Nh=4,fillcolor=Black,Nmr=0](generated_cell12)(22,6){}
\node[Nw=4,Nh=4,fillcolor=Black,Nmr=0](generated_cell13)(26,6){}
\node[Nw=4,Nh=4,fillcolor=Black,Nmr=0](generated_cell14)(30,6){}
\node[Nw=4,Nh=4,fillcolor=Black,Nmr=0](generated_cell15)(34,6){}
\node[Nw=4,Nh=4,fillcolor=Black,Nmr=0](generated_cell16)(22,2){}
\node[Nw=4,Nh=4,fillcolor=Black,Nmr=0](generated_cell17)(26,2){}
\node[Nw=4,Nh=4,fillcolor=Black,Nmr=0](generated_cell18)(30,2){}
\multiput(4, 0)(4,0){12}{\color{DarkGray} \line(0,1){20}}
\multiput(0, 16)(0,-4){5}{\color{DarkGray} \line(1,0){48}}
\end{picture}
    \label{fig:NonogramSignalSol1}}
  \subfigure[]{
    \begin{picture}(48, 20)
\linethickness{0.05mm}
\node[Nw=4,Nh=4,linecolor=White](generated_hnode0)(2,14){3}
\node[Nw=4,Nh=4,linecolor=White](generated_hnode1)(2,10){4}
\node[Nw=4,Nh=4,linecolor=White](generated_hnode2)(2,6){4}
\node[Nw=4,Nh=4,linecolor=White](generated_hnode3)(2,2){3}
\node[Nw=4,Nh=4,linecolor=White](generated_vnode0)(6,18){0}
\node[Nw=4,Nh=4,linecolor=White](generated_vnode1)(10,18){0}
\node[Nw=4,Nh=4,linecolor=White](generated_vnode2)(14,18){0}
\node[Nw=4,Nh=4,linecolor=White](generated_vnode3)(18,18){1}
\node[Nw=4,Nh=4,linecolor=White](generated_vnode4)(22,18){4}
\node[Nw=4,Nh=4,linecolor=White](generated_vnode5)(26,18){4}
\node[Nw=4,Nh=4,linecolor=White](generated_vnode6)(30,18){4}
\node[Nw=4,Nh=4,linecolor=White](generated_vnode7)(34,18){1}
\node[Nw=4,Nh=4,linecolor=White](generated_vnode8)(38,18){0}
\node[Nw=4,Nh=4,linecolor=White](generated_vnode9)(42,18){0}
\node[Nw=4,Nh=4,linecolor=White](generated_vnode10)(46,18){0}
\node[Nw=4,Nh=4,fillcolor=Black,Nmr=0](generated_cell4)(22,14){}
\node[Nw=4,Nh=4,fillcolor=Black,Nmr=0](generated_cell5)(26,14){}
\node[Nw=4,Nh=4,fillcolor=Black,Nmr=0](generated_cell6)(30,14){}
\node[Nw=4,Nh=4,fillcolor=Black,Nmr=0](generated_cell8)(22,10){}
\node[Nw=4,Nh=4,fillcolor=Black,Nmr=0](generated_cell9)(26,10){}
\node[Nw=4,Nh=4,fillcolor=Black,Nmr=0](generated_cell10)(30,10){}
\node[Nw=4,Nh=4,fillcolor=Black,Nmr=0](generated_cell11)(34,10){}
\node[Nw=4,Nh=4,fillcolor=Black,Nmr=0](generated_cell11)(18,6){}
\node[Nw=4,Nh=4,fillcolor=Black,Nmr=0](generated_cell12)(22,6){}
\node[Nw=4,Nh=4,fillcolor=Black,Nmr=0](generated_cell13)(26,6){}
\node[Nw=4,Nh=4,fillcolor=Black,Nmr=0](generated_cell14)(30,6){}
\node[Nw=4,Nh=4,fillcolor=Black,Nmr=0](generated_cell16)(22,2){}
\node[Nw=4,Nh=4,fillcolor=Black,Nmr=0](generated_cell17)(26,2){}
\node[Nw=4,Nh=4,fillcolor=Black,Nmr=0](generated_cell18)(30,2){}
\multiput(4, 0)(4,0){12}{\color{DarkGray} \line(0,1){20}}
\multiput(0, 16)(0,-4){5}{\color{DarkGray} \line(1,0){48}}
\end{picture}
    \label{fig:NonogramSignalSol2}}
  \caption[Nonogram with two solutions]{The two solutions of a Nonogram featuring two subnonograms horizontally separated by $3$ separation lines.}
  \label{fig:NonogramSignal}
\end{figure}
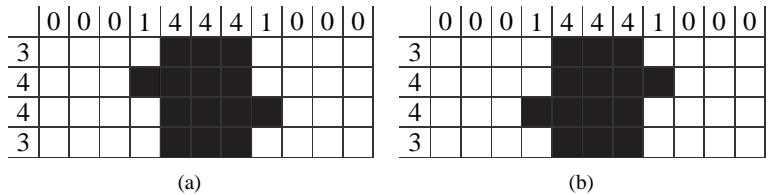

It is also possible to send a signal between two orthogonal adjacent subnonograms by slightly adjusting delimiters between them. This is illustrated in Figure~\ref{fig:NonogramSignal}. The figure shows two subnonograms separated by $D = 3$ separation lines; one subnonogram between cells $(1,1)$ and $(4,4)$, inclusive, and one subnonogram between cells $(1,8)$ and $(4,11)$, inclusive. If we were to decide that $(3,4)$ should be $\mathit{black}$, this would explicitly mean that cell $(3,8)$ cannot be $\mathit{black}$. (Note that this would also explicitly mean that $(2,8)$ is $\mathit{black}$, and $(2,4)$ is not $\mathit{black}$.) The opposite is also true. We will use this property to construct gadgets within a subnonogram, and propagate signals between them. This way we can embed a constraint graph on a Nonogram grid.

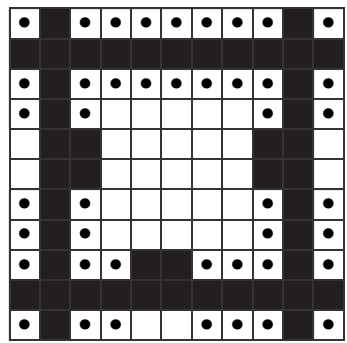
\begin{figure}[!ht]
 \centering
 \begin{picture}(44, 44)
\linethickness{0.1mm}
\node[Nw=4,Nh=4,fillcolor=White,Nmr=0](generated_cellOverlay0)(2,42){$\bullet$}
\node[Nw=4,Nh=4,fillcolor=Black,Nmr=0](generated_cell1)(6,42){}
\node[Nw=4,Nh=4,fillcolor=White,Nmr=0](generated_cellOverlay2)(10,42){$\bullet$}
\node[Nw=4,Nh=4,fillcolor=White,Nmr=0](generated_cellOverlay3)(14,42){$\bullet$}
\node[Nw=4,Nh=4,fillcolor=White,Nmr=0](generated_cellOverlay4)(18,42){$\bullet$}
\node[Nw=4,Nh=4,fillcolor=White,Nmr=0](generated_cellOverlay5)(22,42){$\bullet$}
\node[Nw=4,Nh=4,fillcolor=White,Nmr=0](generated_cellOverlay6)(26,42){$\bullet$}
\node[Nw=4,Nh=4,fillcolor=White,Nmr=0](generated_cellOverlay7)(30,42){$\bullet$}
\node[Nw=4,Nh=4,fillcolor=White,Nmr=0](generated_cellOverlay8)(34,42){$\bullet$}
\node[Nw=4,Nh=4,fillcolor=Black,Nmr=0](generated_cell9)(38,42){}
\node[Nw=4,Nh=4,fillcolor=White,Nmr=0](generated_cellOverlay10)(42,42){$\bullet$}
\node[Nw=4,Nh=4,fillcolor=Black,Nmr=0](generated_cell11)(2,38){}
\node[Nw=4,Nh=4,fillcolor=Black,Nmr=0](generated_cell12)(6,38){}
\node[Nw=4,Nh=4,fillcolor=Black,Nmr=0](generated_cell13)(10,38){}
\node[Nw=4,Nh=4,fillcolor=Black,Nmr=0](generated_cell14)(14,38){}
\node[Nw=4,Nh=4,fillcolor=Black,Nmr=0](generated_cell15)(18,38){}
\node[Nw=4,Nh=4,fillcolor=Black,Nmr=0](generated_cell16)(22,38){}
\node[Nw=4,Nh=4,fillcolor=Black,Nmr=0](generated_cell17)(26,38){}
\node[Nw=4,Nh=4,fillcolor=Black,Nmr=0](generated_cell18)(30,38){}
\node[Nw=4,Nh=4,fillcolor=Black,Nmr=0](generated_cell19)(34,38){}
\node[Nw=4,Nh=4,fillcolor=Black,Nmr=0](generated_cell20)(38,38){}
\node[Nw=4,Nh=4,fillcolor=Black,Nmr=0](generated_cell21)(42,38){}
\node[Nw=4,Nh=4,fillcolor=White,Nmr=0](generated_cellOverlay22)(2,34){$\bullet$}
\node[Nw=4,Nh=4,fillcolor=Black,Nmr=0](generated_cell23)(6,34){}
\node[Nw=4,Nh=4,fillcolor=White,Nmr=0](generated_cellOverlay24)(10,34){$\bullet$}
\node[Nw=4,Nh=4,fillcolor=White,Nmr=0](generated_cellOverlay25)(14,34){$\bullet$}
\node[Nw=4,Nh=4,fillcolor=White,Nmr=0](generated_cellOverlay26)(18,34){$\bullet$}
\node[Nw=4,Nh=4,fillcolor=White,Nmr=0](generated_cellOverlay27)(22,34){$\bullet$}
\node[Nw=4,Nh=4,fillcolor=White,Nmr=0](generated_cellOverlay28)(26,34){$\bullet$}
\node[Nw=4,Nh=4,fillcolor=White,Nmr=0](generated_cellOverlay29)(30,34){$\bullet$}
\node[Nw=4,Nh=4,fillcolor=White,Nmr=0](generated_cellOverlay30)(34,34){$\bullet$}
\node[Nw=4,Nh=4,fillcolor=Black,Nmr=0](generated_cell31)(38,34){}
\node[Nw=4,Nh=4,fillcolor=White,Nmr=0](generated_cellOverlay32)(42,34){$\bullet$}
\node[Nw=4,Nh=4,fillcolor=White,Nmr=0](generated_cellOverlay33)(2,30){$\bullet$}
\node[Nw=4,Nh=4,fillcolor=Black,Nmr=0](generated_cell34)(6,30){}
\node[Nw=4,Nh=4,fillcolor=White,Nmr=0](generated_cellOverlay35)(10,30){$\bullet$}
\node[Nw=4,Nh=4,fillcolor=White,Nmr=0](generated_cellOverlay41)(34,30){$\bullet$}
\node[Nw=4,Nh=4,fillcolor=Black,Nmr=0](generated_cell42)(38,30){}
\node[Nw=4,Nh=4,fillcolor=White,Nmr=0](generated_cellOverlay43)(42,30){$\bullet$}
\node[Nw=4,Nh=4,fillcolor=Black,Nmr=0](generated_cell45)(6,26){}
\node[Nw=4,Nh=4,fillcolor=Black,Nmr=0](generated_cell46)(10,26){}
\node[Nw=4,Nh=4,fillcolor=Black,Nmr=0](generated_cell52)(34,26){}
\node[Nw=4,Nh=4,fillcolor=Black,Nmr=0](generated_cell53)(38,26){}
\node[Nw=4,Nh=4,fillcolor=Black,Nmr=0](generated_cell56)(6,22){}
\node[Nw=4,Nh=4,fillcolor=Black,Nmr=0](generated_cell57)(10,22){}
\node[Nw=4,Nh=4,fillcolor=Black,Nmr=0](generated_cell63)(34,22){}
\node[Nw=4,Nh=4,fillcolor=Black,Nmr=0](generated_cell64)(38,22){}
\node[Nw=4,Nh=4,fillcolor=White,Nmr=0](generated_cellOverlay66)(2,18){$\bullet$}
\node[Nw=4,Nh=4,fillcolor=Black,Nmr=0](generated_cell67)(6,18){}
\node[Nw=4,Nh=4,fillcolor=White,Nmr=0](generated_cellOverlay68)(10,18){$\bullet$}
\node[Nw=4,Nh=4,fillcolor=White,Nmr=0](generated_cellOverlay74)(34,18){$\bullet$}
\node[Nw=4,Nh=4,fillcolor=Black,Nmr=0](generated_cell75)(38,18){}
\node[Nw=4,Nh=4,fillcolor=White,Nmr=0](generated_cellOverlay76)(42,18){$\bullet$}
\node[Nw=4,Nh=4,fillcolor=White,Nmr=0](generated_cellOverlay77)(2,14){$\bullet$}
\node[Nw=4,Nh=4,fillcolor=Black,Nmr=0](generated_cell78)(6,14){}
\node[Nw=4,Nh=4,fillcolor=White,Nmr=0](generated_cellOverlay79)(10,14){$\bullet$}
\node[Nw=4,Nh=4,fillcolor=White,Nmr=0](generated_cellOverlay85)(34,14){$\bullet$}
\node[Nw=4,Nh=4,fillcolor=Black,Nmr=0](generated_cell86)(38,14){}
\node[Nw=4,Nh=4,fillcolor=White,Nmr=0](generated_cellOverlay87)(42,14){$\bullet$}
\node[Nw=4,Nh=4,fillcolor=White,Nmr=0](generated_cellOverlay88)(2,10){$\bullet$}
\node[Nw=4,Nh=4,fillcolor=Black,Nmr=0](generated_cell89)(6,10){}
\node[Nw=4,Nh=4,fillcolor=White,Nmr=0](generated_cellOverlay90)(10,10){$\bullet$}
\node[Nw=4,Nh=4,fillcolor=White,Nmr=0](generated_cellOverlay91)(14,10){$\bullet$}
\node[Nw=4,Nh=4,fillcolor=Black,Nmr=0](generated_cell92)(18,10){}
\node[Nw=4,Nh=4,fillcolor=Black,Nmr=0](generated_cell93)(22,10){}
\node[Nw=4,Nh=4,fillcolor=White,Nmr=0](generated_cellOverlay94)(26,10){$\bullet$}
\node[Nw=4,Nh=4,fillcolor=White,Nmr=0](generated_cellOverlay95)(30,10){$\bullet$}
\node[Nw=4,Nh=4,fillcolor=White,Nmr=0](generated_cellOverlay96)(34,10){$\bullet$}
\node[Nw=4,Nh=4,fillcolor=Black,Nmr=0](generated_cell97)(38,10){}
\node[Nw=4,Nh=4,fillcolor=White,Nmr=0](generated_cellOverlay98)(42,10){$\bullet$}
\node[Nw=4,Nh=4,fillcolor=Black,Nmr=0](generated_cell99)(2,6){}
\node[Nw=4,Nh=4,fillcolor=Black,Nmr=0](generated_cell100)(6,6){}
\node[Nw=4,Nh=4,fillcolor=Black,Nmr=0](generated_cell101)(10,6){}
\node[Nw=4,Nh=4,fillcolor=Black,Nmr=0](generated_cell102)(14,6){}
\node[Nw=4,Nh=4,fillcolor=Black,Nmr=0](generated_cell103)(18,6){}
\node[Nw=4,Nh=4,fillcolor=Black,Nmr=0](generated_cell104)(22,6){}
\node[Nw=4,Nh=4,fillcolor=Black,Nmr=0](generated_cell105)(26,6){}
\node[Nw=4,Nh=4,fillcolor=Black,Nmr=0](generated_cell106)(30,6){}
\node[Nw=4,Nh=4,fillcolor=Black,Nmr=0](generated_cell107)(34,6){}
\node[Nw=4,Nh=4,fillcolor=Black,Nmr=0](generated_cell108)(38,6){}
\node[Nw=4,Nh=4,fillcolor=Black,Nmr=0](generated_cell109)(42,6){}
\node[Nw=4,Nh=4,fillcolor=White,Nmr=0](generated_cellOverlay110)(2,2){$\bullet$}
\node[Nw=4,Nh=4,fillcolor=Black,Nmr=0](generated_cell111)(6,2){}
\node[Nw=4,Nh=4,fillcolor=White,Nmr=0](generated_cellOverlay112)(10,2){$\bullet$}
\node[Nw=4,Nh=4,fillcolor=White,Nmr=0](generated_cellOverlay113)(14,2){$\bullet$}
\node[Nw=4,Nh=4,fillcolor=White,Nmr=0](generated_cellOverlay116)(26,2){$\bullet$}
\node[Nw=4,Nh=4,fillcolor=White,Nmr=0](generated_cellOverlay117)(30,2){$\bullet$}
\node[Nw=4,Nh=4,fillcolor=White,Nmr=0](generated_cellOverlay118)(34,2){$\bullet$}
\node[Nw=4,Nh=4,fillcolor=Black,Nmr=0](generated_cell119)(38,2){}
\node[Nw=4,Nh=4,fillcolor=White,Nmr=0](generated_cellOverlay120)(42,2){$\bullet$}
\multiput(0, 0)(4,0){12}{\color{DarkGray} \line(0,1){44}}
\multiput(0, 44)(0,-4){12}{\color{DarkGray} \line(1,0){44}}
\end{picture}
 \caption[Template for Nonogram gadgets]{Template for Nonogram gadgets. Squares that are black are black in all instances; squares with a dot are white in all instances.}
 \label{fig:nonogramGadgetOverall}
\end{figure}

In Figure~\ref{fig:nonogramGadgetOverall} a template of the gadgets is shown. From the description of every gadget follows immediately that the black cells must be $\mathit{black}$ and the dotted cells must be $\mathit{white}$. The state of the other cells is dependent on the type and state of the gadget.

The gadgets are shown in Figure~\ref{fig:NonogramGadgets} and have $G = 7$. As it does not influence the functionality, on each side these are surrounded by only one separation line. (In the large construction we will use $D = 5$, with obvious adaptations of the descriptions.) These are already $\mathit{black}$. If a cell corresponding to an edge is $\mathit{white}$, this means that the edge is pointing away from the vertex, and vice versa. 
Besides AND and OR gadgets, we also provide two gadgets needed for wiring.
\begin{figure}[!ht]
 \centering
  \subfigure[AND gadget]{
    \begin{picture}(60, 60)
\linethickness{0.1mm}
\node[Nw=4,Nh=4,linecolor=White](generated_hnode0)(10,42){$1$}
\node[Nw=4,Nh=4,linecolor=White](generated_hnode1)(14,42){$1$}
\node[Nw=4,Nh=4,linecolor=White](generated_hnode4)(14,38){$1\!1$}
\node[Nw=4,Nh=4,linecolor=White](generated_hnode8)(10,34){$1$}
\node[Nw=4,Nh=4,linecolor=White](generated_hnode9)(14,34){$1$}
\node[Nw=4,Nh=4,linecolor=White](generated_hnode12)(2,30){}
\node[Nw=4,Nh=4,linecolor=White](generated_hnode13)(6,30){$1$}
\node[Nw=4,Nh=4,linecolor=White](generated_hnode14)(10,30){$3$}
\node[Nw=4,Nh=4,linecolor=White](generated_hnode15)(14,30){$1$}
\node[Nw=4,Nh=4,linecolor=White](generated_hnode16)(6,26){$3$}
\node[Nw=4,Nh=4,linecolor=White](generated_hnode17)(10,26){$1$}
\node[Nw=4,Nh=4,linecolor=White](generated_hnode18)(14,26){$3$}
\node[Nw=4,Nh=4,linecolor=White](generated_hnode20)(6,22){$3$}
\node[Nw=4,Nh=4,linecolor=White](generated_hnode21)(10,22){$2$}
\node[Nw=4,Nh=4,linecolor=White](generated_hnode22)(14,22){$3$}
\node[Nw=4,Nh=4,linecolor=White](generated_hnode24)(2,18){}
\node[Nw=4,Nh=4,linecolor=White](generated_hnode25)(6,18){$1$}
\node[Nw=4,Nh=4,linecolor=White](generated_hnode26)(10,18){$2$}
\node[Nw=4,Nh=4,linecolor=White](generated_hnode27)(14,18){$1$}
\node[Nw=4,Nh=4,linecolor=White](generated_hnode28)(2,14){$1$}
\node[Nw=4,Nh=4,linecolor=White](generated_hnode29)(6,14){$1$}
\node[Nw=4,Nh=4,linecolor=White](generated_hnode30)(10,14){$1$}
\node[Nw=4,Nh=4,linecolor=White](generated_hnode31)(14,14){$1$}
\node[Nw=4,Nh=4,linecolor=White](generated_hnode32)(6,10){$1$}
\node[Nw=4,Nh=4,linecolor=White](generated_hnode33)(10,10){$2$}
\node[Nw=4,Nh=4,linecolor=White](generated_hnode34)(14,10){$1$}
\node[Nw=4,Nh=4,linecolor=White](generated_hnode36)(14,6){$1\!1$}
\node[Nw=4,Nh=4,linecolor=White](generated_hnode40)(6,2){$1$}
\node[Nw=4,Nh=4,linecolor=White](generated_hnode41)(10,2){$1$}
\node[Nw=4,Nh=4,linecolor=White](generated_hnode42)(14,2){$1$}
\node[Nw=4,Nh=4,linecolor=White](generated_vnode0)(18,54){$1$}
\node[Nw=4,Nh=4,linecolor=White](generated_vnode1)(18,50){$1$}
\node[Nw=4,Nh=4,linecolor=White](generated_vnode2)(18,46){$1$}
\node[Nw=4,Nh=4,linecolor=White](generated_vnode4)(22,46){$1\!1$}
\node[Nw=4,Nh=4,linecolor=White](generated_vnode8)(26,54){$1$}
\node[Nw=4,Nh=4,linecolor=White](generated_vnode9)(26,50){$2$}
\node[Nw=4,Nh=4,linecolor=White](generated_vnode10)(26,46){$1$}
\node[Nw=4,Nh=4,linecolor=White](generated_vnode12)(30,54){$1$}
\node[Nw=4,Nh=4,linecolor=White](generated_vnode13)(30,50){$1$}
\node[Nw=4,Nh=4,linecolor=White](generated_vnode14)(30,46){$1$}
\node[Nw=4,Nh=4,linecolor=White](generated_vnode16)(34,54){$1$}
\node[Nw=4,Nh=4,linecolor=White](generated_vnode17)(34,50){$1$}
\node[Nw=4,Nh=4,linecolor=White](generated_vnode18)(34,46){$3$}
\node[Nw=4,Nh=4,linecolor=White](generated_vnode20)(38,54){$1$}
\node[Nw=4,Nh=4,linecolor=White](generated_vnode21)(38,50){$3$}
\node[Nw=4,Nh=4,linecolor=White](generated_vnode22)(38,46){$3$}
\node[Nw=4,Nh=4,linecolor=White](generated_vnode24)(42,58){$1$}
\node[Nw=4,Nh=4,linecolor=White](generated_vnode25)(42,54){$1$}
\node[Nw=4,Nh=4,linecolor=White](generated_vnode26)(42,50){$2$}
\node[Nw=4,Nh=4,linecolor=White](generated_vnode27)(42,46){$1$}
\node[Nw=4,Nh=4,linecolor=White](generated_vnode28)(46,58){$1$}
\node[Nw=4,Nh=4,linecolor=White](generated_vnode29)(46,54){$1$}
\node[Nw=4,Nh=4,linecolor=White](generated_vnode30)(46,50){$2$}
\node[Nw=4,Nh=4,linecolor=White](generated_vnode31)(46,46){$1$}
\node[Nw=4,Nh=4,linecolor=White](generated_vnode32)(50,54){$1$}
\node[Nw=4,Nh=4,linecolor=White](generated_vnode33)(50,50){$2$}
\node[Nw=4,Nh=4,linecolor=White](generated_vnode34)(50,46){$1$}
\node[Nw=4,Nh=4,linecolor=White](generated_vnode36)(54,46){$1\!1$}
\node[Nw=4,Nh=4,linecolor=White](generated_vnode40)(58,54){$1$}
\node[Nw=4,Nh=4,linecolor=White](generated_vnode41)(58,50){$1$}
\node[Nw=4,Nh=4,linecolor=White](generated_vnode42)(58,46){$1$}
\node[Nw=4,Nh=4,fillcolor=Black,Nmr=0](generated_cell1)(22,42){}
\node[Nw=4,Nh=4,fillcolor=Black,Nmr=0](generated_cell9)(54,42){}
\node[Nw=4,Nh=4,fillcolor=Black,Nmr=0](generated_cell11)(18,38){}
\node[Nw=4,Nh=4,fillcolor=Black,Nmr=0](generated_cell12)(22,38){}
\node[Nw=4,Nh=4,fillcolor=Black,Nmr=0](generated_cell13)(26,38){}
\node[Nw=4,Nh=4,fillcolor=Black,Nmr=0](generated_cell14)(30,38){}
\node[Nw=4,Nh=4,fillcolor=Black,Nmr=0](generated_cell15)(34,38){}
\node[Nw=4,Nh=4,fillcolor=Black,Nmr=0](generated_cell16)(38,38){}
\node[Nw=4,Nh=4,fillcolor=Black,Nmr=0](generated_cell17)(42,38){}
\node[Nw=4,Nh=4,fillcolor=Black,Nmr=0](generated_cell18)(46,38){}
\node[Nw=4,Nh=4,fillcolor=Black,Nmr=0](generated_cell19)(50,38){}
\node[Nw=4,Nh=4,fillcolor=Black,Nmr=0](generated_cell20)(54,38){}
\node[Nw=4,Nh=4,fillcolor=Black,Nmr=0](generated_cell21)(58,38){}
\node[Nw=4,Nh=4,fillcolor=Black,Nmr=0](generated_cell23)(22,34){}
\node[Nw=4,Nh=4,fillcolor=Black,Nmr=0](generated_cell31)(54,34){}
\node[Nw=4,Nh=4,fillcolor=Black,Nmr=0](generated_cell34)(22,30){}
\node[Nw=4,Nh=4,fillcolor=Black,Nmr=0](generated_cell42)(54,30){}
\node[Nw=4,Nh=4,fillcolor=Black,Nmr=0](generated_cell45)(22,26){}
\node[Nw=4,Nh=4,fillcolor=Black,Nmr=0](generated_cell53)(54,26){}
\node[Nw=4,Nh=4,fillcolor=White,Nmr=0](generated_cellLetter55)(18,22){\footnotesize${\sf a}$}
\node[Nw=4,Nh=4,fillcolor=Black,Nmr=0](generated_cell56)(22,22){}
\node[Nw=4,Nh=4,fillcolor=Black,Nmr=0](generated_cell64)(54,22){}
\node[Nw=4,Nh=4,fillcolor=White,Nmr=0](generated_cellLetter65)(58,22){\footnotesize${\sf c}$}
\node[Nw=4,Nh=4,fillcolor=Black,Nmr=0](generated_cell67)(22,18){}
\node[Nw=4,Nh=4,fillcolor=Black,Nmr=0](generated_cell75)(54,18){}
\node[Nw=4,Nh=4,fillcolor=Black,Nmr=0](generated_cell78)(22,14){}
\node[Nw=4,Nh=4,fillcolor=Black,Nmr=0](generated_cell86)(54,14){}
\node[Nw=4,Nh=4,fillcolor=Black,Nmr=0](generated_cell89)(22,10){}
\node[Nw=4,Nh=4,fillcolor=Black,Nmr=0](generated_cell97)(54,10){}
\node[Nw=4,Nh=4,fillcolor=Black,Nmr=0](generated_cell99)(18,6){}
\node[Nw=4,Nh=4,fillcolor=Black,Nmr=0](generated_cell100)(22,6){}
\node[Nw=4,Nh=4,fillcolor=Black,Nmr=0](generated_cell101)(26,6){}
\node[Nw=4,Nh=4,fillcolor=Black,Nmr=0](generated_cell102)(30,6){}
\node[Nw=4,Nh=4,fillcolor=Black,Nmr=0](generated_cell103)(34,6){}
\node[Nw=4,Nh=4,fillcolor=Black,Nmr=0](generated_cell104)(38,6){}
\node[Nw=4,Nh=4,fillcolor=Black,Nmr=0](generated_cell105)(42,6){}
\node[Nw=4,Nh=4,fillcolor=Black,Nmr=0](generated_cell106)(46,6){}
\node[Nw=4,Nh=4,fillcolor=Black,Nmr=0](generated_cell107)(50,6){}
\node[Nw=4,Nh=4,fillcolor=Black,Nmr=0](generated_cell108)(54,6){}
\node[Nw=4,Nh=4,fillcolor=Black,Nmr=0](generated_cell109)(58,6){}
\node[Nw=4,Nh=4,fillcolor=Black,Nmr=0](generated_cell111)(22,2){}
\node[Nw=4,Nh=4,fillcolor=White,Nmr=0](generated_cellLetter115)(38,2){\footnotesize${\sf b}$}
\node[Nw=4,Nh=4,fillcolor=Black,Nmr=0](generated_cell119)(54,2){}
\multiput(16, 0)(4,0){12}{\color{DarkGray} \line(0,1){60}}
\multiput(0, 44)(0,-4){12}{\color{DarkGray} \line(1,0){60}}
\end{picture}
    \label{fig:NonogramAND}}
  \subfigure[OR gadget]{
    \begin{picture}(60, 60)
\linethickness{0.1mm}
\node[Nw=4,Nh=4,linecolor=White](generated_hnode0)(10,42){$1$}
\node[Nw=4,Nh=4,linecolor=White](generated_hnode1)(14,42){$1$}
\node[Nw=4,Nh=4,linecolor=White](generated_hnode4)(14,38){$1\!1$}
\node[Nw=4,Nh=4,linecolor=White](generated_hnode8)(10,34){$1$}
\node[Nw=4,Nh=4,linecolor=White](generated_hnode9)(14,34){$1$}
\node[Nw=4,Nh=4,linecolor=White](generated_hnode13)(6,30){$1$}
\node[Nw=4,Nh=4,linecolor=White](generated_hnode14)(10,30){$2$}
\node[Nw=4,Nh=4,linecolor=White](generated_hnode15)(14,30){$1$}
\node[Nw=4,Nh=4,linecolor=White](generated_hnode16)(6,26){$3$}
\node[Nw=4,Nh=4,linecolor=White](generated_hnode17)(10,26){$1$}
\node[Nw=4,Nh=4,linecolor=White](generated_hnode18)(14,26){$3$}
\node[Nw=4,Nh=4,linecolor=White](generated_hnode20)(6,22){$3$}
\node[Nw=4,Nh=4,linecolor=White](generated_hnode21)(10,22){$1$}
\node[Nw=4,Nh=4,linecolor=White](generated_hnode22)(14,22){$3$}
\node[Nw=4,Nh=4,linecolor=White](generated_hnode25)(6,18){$1$}
\node[Nw=4,Nh=4,linecolor=White](generated_hnode26)(10,18){$1$}
\node[Nw=4,Nh=4,linecolor=White](generated_hnode27)(14,18){$1$}
\node[Nw=4,Nh=4,linecolor=White](generated_hnode28)(2,14){$1$}
\node[Nw=4,Nh=4,linecolor=White](generated_hnode29)(6,14){$1$}
\node[Nw=4,Nh=4,linecolor=White](generated_hnode30)(10,14){$1$}
\node[Nw=4,Nh=4,linecolor=White](generated_hnode31)(14,14){$1$}
\node[Nw=4,Nh=4,linecolor=White](generated_hnode32)(6,10){$1$}
\node[Nw=4,Nh=4,linecolor=White](generated_hnode33)(10,10){$2$}
\node[Nw=4,Nh=4,linecolor=White](generated_hnode34)(14,10){$1$}
\node[Nw=4,Nh=4,linecolor=White](generated_hnode36)(14,6){$1\!1$}
\node[Nw=4,Nh=4,linecolor=White](generated_hnode40)(6,2){$1$}
\node[Nw=4,Nh=4,linecolor=White](generated_hnode41)(10,2){$1$}
\node[Nw=4,Nh=4,linecolor=White](generated_hnode42)(14,2){$1$}
\node[Nw=4,Nh=4,linecolor=White](generated_vnode0)(18,54){$1$}
\node[Nw=4,Nh=4,linecolor=White](generated_vnode1)(18,50){$1$}
\node[Nw=4,Nh=4,linecolor=White](generated_vnode2)(18,46){$1$}
\node[Nw=4,Nh=4,linecolor=White](generated_vnode4)(22,46){$1\!1$}
\node[Nw=4,Nh=4,linecolor=White](generated_vnode8)(26,54){$1$}
\node[Nw=4,Nh=4,linecolor=White](generated_vnode9)(26,50){$2$}
\node[Nw=4,Nh=4,linecolor=White](generated_vnode10)(26,46){$1$}
\node[Nw=4,Nh=4,linecolor=White](generated_vnode11)(30,58){$1$}
\node[Nw=4,Nh=4,linecolor=White](generated_vnode12)(30,54){$1$}
\node[Nw=4,Nh=4,linecolor=White](generated_vnode13)(30,50){$1$}
\node[Nw=4,Nh=4,linecolor=White](generated_vnode14)(30,46){$1$}
\node[Nw=4,Nh=4,linecolor=White](generated_vnode16)(34,54){$1$}
\node[Nw=4,Nh=4,linecolor=White](generated_vnode17)(34,50){$1$}
\node[Nw=4,Nh=4,linecolor=White](generated_vnode18)(34,46){$3$}
\node[Nw=4,Nh=4,linecolor=White](generated_vnode20)(38,54){$1$}
\node[Nw=4,Nh=4,linecolor=White](generated_vnode21)(38,50){$3$}
\node[Nw=4,Nh=4,linecolor=White](generated_vnode22)(38,46){$3$}
\node[Nw=4,Nh=4,linecolor=White](generated_vnode25)(42,50){$1$}
\node[Nw=4,Nh=4,linecolor=White](generated_vnode27)(42,46){$1$}
\node[Nw=4,Nh=4,linecolor=White](generated_vnode28)(46,58){$1$}
\node[Nw=4,Nh=4,linecolor=White](generated_vnode29)(46,54){$1$}
\node[Nw=4,Nh=4,linecolor=White](generated_vnode30)(46,50){$1$}
\node[Nw=4,Nh=4,linecolor=White](generated_vnode31)(46,46){$1$}
\node[Nw=4,Nh=4,linecolor=White](generated_vnode32)(50,54){$1$}
\node[Nw=4,Nh=4,linecolor=White](generated_vnode33)(50,50){$2$}
\node[Nw=4,Nh=4,linecolor=White](generated_vnode34)(50,46){$1$}
\node[Nw=4,Nh=4,linecolor=White](generated_vnode36)(54,46){$1\!1$}
\node[Nw=4,Nh=4,linecolor=White](generated_vnode40)(58,54){$1$}
\node[Nw=4,Nh=4,linecolor=White](generated_vnode41)(58,50){$1$}
\node[Nw=4,Nh=4,linecolor=White](generated_vnode42)(58,46){$1$}
\node[Nw=4,Nh=4,fillcolor=Black,Nmr=0](generated_cell1)(22,42){}
\node[Nw=4,Nh=4,fillcolor=Black,Nmr=0](generated_cell9)(54,42){}
\node[Nw=4,Nh=4,fillcolor=Black,Nmr=0](generated_cell11)(18,38){}
\node[Nw=4,Nh=4,fillcolor=Black,Nmr=0](generated_cell12)(22,38){}
\node[Nw=4,Nh=4,fillcolor=Black,Nmr=0](generated_cell13)(26,38){}
\node[Nw=4,Nh=4,fillcolor=Black,Nmr=0](generated_cell14)(30,38){}
\node[Nw=4,Nh=4,fillcolor=Black,Nmr=0](generated_cell15)(34,38){}
\node[Nw=4,Nh=4,fillcolor=Black,Nmr=0](generated_cell16)(38,38){}
\node[Nw=4,Nh=4,fillcolor=Black,Nmr=0](generated_cell17)(42,38){}
\node[Nw=4,Nh=4,fillcolor=Black,Nmr=0](generated_cell18)(46,38){}
\node[Nw=4,Nh=4,fillcolor=Black,Nmr=0](generated_cell19)(50,38){}
\node[Nw=4,Nh=4,fillcolor=Black,Nmr=0](generated_cell20)(54,38){}
\node[Nw=4,Nh=4,fillcolor=Black,Nmr=0](generated_cell21)(58,38){}
\node[Nw=4,Nh=4,fillcolor=Black,Nmr=0](generated_cell23)(22,34){}
\node[Nw=4,Nh=4,fillcolor=Black,Nmr=0](generated_cell31)(54,34){}
\node[Nw=4,Nh=4,fillcolor=Black,Nmr=0](generated_cell34)(22,30){}
\node[Nw=4,Nh=4,fillcolor=Black,Nmr=0](generated_cell42)(54,30){}
\node[Nw=4,Nh=4,fillcolor=Black,Nmr=0](generated_cell45)(22,26){}
\node[Nw=4,Nh=4,fillcolor=Black,Nmr=0](generated_cell53)(54,26){}
\node[Nw=4,Nh=4,fillcolor=White,Nmr=0](generated_cellLetter55)(18,22){\footnotesize${\sf a}$}
\node[Nw=4,Nh=4,fillcolor=Black,Nmr=0](generated_cell56)(22,22){}
\node[Nw=4,Nh=4,fillcolor=Black,Nmr=0](generated_cell64)(54,22){}
\node[Nw=4,Nh=4,fillcolor=White,Nmr=0](generated_cellLetter65)(58,22){\footnotesize${\sf c}$}
\node[Nw=4,Nh=4,fillcolor=Black,Nmr=0](generated_cell67)(22,18){}
\node[Nw=4,Nh=4,fillcolor=Black,Nmr=0](generated_cell75)(54,18){}
\node[Nw=4,Nh=4,fillcolor=Black,Nmr=0](generated_cell78)(22,14){}
\node[Nw=4,Nh=4,fillcolor=Black,Nmr=0](generated_cell86)(54,14){}
\node[Nw=4,Nh=4,fillcolor=Black,Nmr=0](generated_cell89)(22,10){}
\node[Nw=4,Nh=4,fillcolor=Black,Nmr=0](generated_cell97)(54,10){}
\node[Nw=4,Nh=4,fillcolor=Black,Nmr=0](generated_cell99)(18,6){}
\node[Nw=4,Nh=4,fillcolor=Black,Nmr=0](generated_cell100)(22,6){}
\node[Nw=4,Nh=4,fillcolor=Black,Nmr=0](generated_cell101)(26,6){}
\node[Nw=4,Nh=4,fillcolor=Black,Nmr=0](generated_cell102)(30,6){}
\node[Nw=4,Nh=4,fillcolor=Black,Nmr=0](generated_cell103)(34,6){}
\node[Nw=4,Nh=4,fillcolor=Black,Nmr=0](generated_cell104)(38,6){}
\node[Nw=4,Nh=4,fillcolor=Black,Nmr=0](generated_cell105)(42,6){}
\node[Nw=4,Nh=4,fillcolor=Black,Nmr=0](generated_cell106)(46,6){}
\node[Nw=4,Nh=4,fillcolor=Black,Nmr=0](generated_cell107)(50,6){}
\node[Nw=4,Nh=4,fillcolor=Black,Nmr=0](generated_cell108)(54,6){}
\node[Nw=4,Nh=4,fillcolor=Black,Nmr=0](generated_cell109)(58,6){}
\node[Nw=4,Nh=4,fillcolor=Black,Nmr=0](generated_cell111)(22,2){}
\node[Nw=4,Nh=4,fillcolor=White,Nmr=0](generated_cellLetter115)(38,2){\footnotesize${\sf b}$}
\node[Nw=4,Nh=4,fillcolor=Black,Nmr=0](generated_cell119)(54,2){}
\multiput(16, 0)(4,0){12}{\color{DarkGray} \line(0,1){60}}
\multiput(0, 44)(0,-4){12}{\color{DarkGray} \line(1,0){60}}
\end{picture}
    \label{fig:NonogramOR}}
  \subfigure[Wire, turn]{
    \begin{picture}(60, 60)
\linethickness{0.1mm}
\node[Nw=4,Nh=4,linecolor=White](generated_hnode0)(10,42){$1$}
\node[Nw=4,Nh=4,linecolor=White](generated_hnode1)(14,42){$1$}
\node[Nw=4,Nh=4,linecolor=White](generated_hnode4)(14,38){$1\!1$}
\node[Nw=4,Nh=4,linecolor=White](generated_hnode8)(10,34){$1$}
\node[Nw=4,Nh=4,linecolor=White](generated_hnode9)(14,34){$1$}
\node[Nw=4,Nh=4,linecolor=White](generated_hnode14)(10,30){$1$}
\node[Nw=4,Nh=4,linecolor=White](generated_hnode15)(14,30){$1$}
\node[Nw=4,Nh=4,linecolor=White](generated_hnode16)(6,26){$3$}
\node[Nw=4,Nh=4,linecolor=White](generated_hnode17)(10,26){$1$}
\node[Nw=4,Nh=4,linecolor=White](generated_hnode18)(14,26){$1$}
\node[Nw=4,Nh=4,linecolor=White](generated_hnode20)(6,22){$3$}
\node[Nw=4,Nh=4,linecolor=White](generated_hnode21)(10,22){$1$}
\node[Nw=4,Nh=4,linecolor=White](generated_hnode22)(14,22){$1$}
\node[Nw=4,Nh=4,linecolor=White](generated_hnode25)(6,18){$1$}
\node[Nw=4,Nh=4,linecolor=White](generated_hnode26)(10,18){$1$}
\node[Nw=4,Nh=4,linecolor=White](generated_hnode27)(14,18){$1$}
\node[Nw=4,Nh=4,linecolor=White](generated_hnode29)(6,14){$1$}
\node[Nw=4,Nh=4,linecolor=White](generated_hnode30)(10,14){$2$}
\node[Nw=4,Nh=4,linecolor=White](generated_hnode31)(14,14){$1$}
\node[Nw=4,Nh=4,linecolor=White](generated_hnode32)(6,10){$1$}
\node[Nw=4,Nh=4,linecolor=White](generated_hnode33)(10,10){$2$}
\node[Nw=4,Nh=4,linecolor=White](generated_hnode34)(14,10){$1$}
\node[Nw=4,Nh=4,linecolor=White](generated_hnode36)(14,6){$1\!1$}
\node[Nw=4,Nh=4,linecolor=White](generated_hnode40)(6,2){$1$}
\node[Nw=4,Nh=4,linecolor=White](generated_hnode41)(10,2){$1$}
\node[Nw=4,Nh=4,linecolor=White](generated_hnode42)(14,2){$1$}
\node[Nw=4,Nh=4,linecolor=White](generated_vnode0)(18,54){$1$}
\node[Nw=4,Nh=4,linecolor=White](generated_vnode1)(18,50){$1$}
\node[Nw=4,Nh=4,linecolor=White](generated_vnode2)(18,46){$1$}
\node[Nw=4,Nh=4,linecolor=White](generated_vnode4)(22,46){$1\!1$}
\node[Nw=4,Nh=4,linecolor=White](generated_vnode8)(26,54){$1$}
\node[Nw=4,Nh=4,linecolor=White](generated_vnode9)(26,50){$2$}
\node[Nw=4,Nh=4,linecolor=White](generated_vnode10)(26,46){$1$}
\node[Nw=4,Nh=4,linecolor=White](generated_vnode12)(30,54){$1$}
\node[Nw=4,Nh=4,linecolor=White](generated_vnode13)(30,50){$1$}
\node[Nw=4,Nh=4,linecolor=White](generated_vnode14)(30,46){$1$}
\node[Nw=4,Nh=4,linecolor=White](generated_vnode16)(34,54){$1$}
\node[Nw=4,Nh=4,linecolor=White](generated_vnode17)(34,50){$1$}
\node[Nw=4,Nh=4,linecolor=White](generated_vnode18)(34,46){$4$}
\node[Nw=4,Nh=4,linecolor=White](generated_vnode20)(38,54){$1$}
\node[Nw=4,Nh=4,linecolor=White](generated_vnode21)(38,50){$1$}
\node[Nw=4,Nh=4,linecolor=White](generated_vnode22)(38,46){$4$}
\node[Nw=4,Nh=4,linecolor=White](generated_vnode26)(42,50){$1$}
\node[Nw=4,Nh=4,linecolor=White](generated_vnode27)(42,46){$1$}
\node[Nw=4,Nh=4,linecolor=White](generated_vnode30)(46,50){$1$}
\node[Nw=4,Nh=4,linecolor=White](generated_vnode31)(46,46){$1$}
\node[Nw=4,Nh=4,linecolor=White](generated_vnode33)(50,50){$1$}
\node[Nw=4,Nh=4,linecolor=White](generated_vnode34)(50,46){$1$}
\node[Nw=4,Nh=4,linecolor=White](generated_vnode36)(54,46){$1\!1$}
\node[Nw=4,Nh=4,linecolor=White](generated_vnode41)(58,50){$1$}
\node[Nw=4,Nh=4,linecolor=White](generated_vnode42)(58,46){$1$}
\node[Nw=4,Nh=4,fillcolor=Black,Nmr=0](generated_cell1)(22,42){}
\node[Nw=4,Nh=4,fillcolor=Black,Nmr=0](generated_cell9)(54,42){}
\node[Nw=4,Nh=4,fillcolor=Black,Nmr=0](generated_cell11)(18,38){}
\node[Nw=4,Nh=4,fillcolor=Black,Nmr=0](generated_cell12)(22,38){}
\node[Nw=4,Nh=4,fillcolor=Black,Nmr=0](generated_cell13)(26,38){}
\node[Nw=4,Nh=4,fillcolor=Black,Nmr=0](generated_cell14)(30,38){}
\node[Nw=4,Nh=4,fillcolor=Black,Nmr=0](generated_cell15)(34,38){}
\node[Nw=4,Nh=4,fillcolor=Black,Nmr=0](generated_cell16)(38,38){}
\node[Nw=4,Nh=4,fillcolor=Black,Nmr=0](generated_cell17)(42,38){}
\node[Nw=4,Nh=4,fillcolor=Black,Nmr=0](generated_cell18)(46,38){}
\node[Nw=4,Nh=4,fillcolor=Black,Nmr=0](generated_cell19)(50,38){}
\node[Nw=4,Nh=4,fillcolor=Black,Nmr=0](generated_cell20)(54,38){}
\node[Nw=4,Nh=4,fillcolor=Black,Nmr=0](generated_cell21)(58,38){}
\node[Nw=4,Nh=4,fillcolor=Black,Nmr=0](generated_cell23)(22,34){}
\node[Nw=4,Nh=4,fillcolor=Black,Nmr=0](generated_cell31)(54,34){}
\node[Nw=4,Nh=4,fillcolor=Black,Nmr=0](generated_cell34)(22,30){}
\node[Nw=4,Nh=4,fillcolor=Black,Nmr=0](generated_cell42)(54,30){}
\node[Nw=4,Nh=4,fillcolor=Black,Nmr=0](generated_cell45)(22,26){}
\node[Nw=4,Nh=4,fillcolor=Black,Nmr=0](generated_cell53)(54,26){}
\node[Nw=4,Nh=4,fillcolor=White,Nmr=0](generated_cellLetter55)(18,22){\footnotesize${\sf a}$}
\node[Nw=4,Nh=4,fillcolor=Black,Nmr=0](generated_cell56)(22,22){}
\node[Nw=4,Nh=4,fillcolor=Black,Nmr=0](generated_cell64)(54,22){}
\node[Nw=4,Nh=4,fillcolor=Black,Nmr=0](generated_cell67)(22,18){}
\node[Nw=4,Nh=4,fillcolor=Black,Nmr=0](generated_cell75)(54,18){}
\node[Nw=4,Nh=4,fillcolor=Black,Nmr=0](generated_cell78)(22,14){}
\node[Nw=4,Nh=4,fillcolor=Black,Nmr=0](generated_cell86)(54,14){}
\node[Nw=4,Nh=4,fillcolor=Black,Nmr=0](generated_cell89)(22,10){}
\node[Nw=4,Nh=4,fillcolor=Black,Nmr=0](generated_cell97)(54,10){}
\node[Nw=4,Nh=4,fillcolor=Black,Nmr=0](generated_cell99)(18,6){}
\node[Nw=4,Nh=4,fillcolor=Black,Nmr=0](generated_cell100)(22,6){}
\node[Nw=4,Nh=4,fillcolor=Black,Nmr=0](generated_cell101)(26,6){}
\node[Nw=4,Nh=4,fillcolor=Black,Nmr=0](generated_cell102)(30,6){}
\node[Nw=4,Nh=4,fillcolor=Black,Nmr=0](generated_cell103)(34,6){}
\node[Nw=4,Nh=4,fillcolor=Black,Nmr=0](generated_cell104)(38,6){}
\node[Nw=4,Nh=4,fillcolor=Black,Nmr=0](generated_cell105)(42,6){}
\node[Nw=4,Nh=4,fillcolor=Black,Nmr=0](generated_cell106)(46,6){}
\node[Nw=4,Nh=4,fillcolor=Black,Nmr=0](generated_cell107)(50,6){}
\node[Nw=4,Nh=4,fillcolor=Black,Nmr=0](generated_cell108)(54,6){}
\node[Nw=4,Nh=4,fillcolor=Black,Nmr=0](generated_cell109)(58,6){}
\node[Nw=4,Nh=4,fillcolor=Black,Nmr=0](generated_cell111)(22,2){}
\node[Nw=4,Nh=4,fillcolor=White,Nmr=0](generated_cellLetter115)(38,2){\footnotesize${\sf b}$}
\node[Nw=4,Nh=4,fillcolor=Black,Nmr=0](generated_cell119)(54,2){}
\multiput(16, 0)(4,0){12}{\color{DarkGray} \line(0,1){60}}
\multiput(0, 44)(0,-4){12}{\color{DarkGray} \line(1,0){60}}
\end{picture}
    \label{fig:NonogramTURN}}
  \subfigure[Wire, straight]{
    \begin{picture}(60, 60)
\linethickness{0.1mm}
\node[Nw=4,Nh=4,linecolor=White](generated_hnode0)(10,42){$1$}
\node[Nw=4,Nh=4,linecolor=White](generated_hnode1)(14,42){$1$}
\node[Nw=4,Nh=4,linecolor=White](generated_hnode4)(14,38){$1\!1$}
\node[Nw=4,Nh=4,linecolor=White](generated_hnode8)(10,34){$1$}
\node[Nw=4,Nh=4,linecolor=White](generated_hnode9)(14,34){$1$}
\node[Nw=4,Nh=4,linecolor=White](generated_hnode14)(10,30){$1$}
\node[Nw=4,Nh=4,linecolor=White](generated_hnode15)(14,30){$1$}
\node[Nw=4,Nh=4,linecolor=White](generated_hnode16)(6,26){$3$}
\node[Nw=4,Nh=4,linecolor=White](generated_hnode17)(10,26){$2$}
\node[Nw=4,Nh=4,linecolor=White](generated_hnode18)(14,26){$3$}
\node[Nw=4,Nh=4,linecolor=White](generated_hnode20)(6,22){$3$}
\node[Nw=4,Nh=4,linecolor=White](generated_hnode21)(10,22){$2$}
\node[Nw=4,Nh=4,linecolor=White](generated_hnode22)(14,22){$3$}
\node[Nw=4,Nh=4,linecolor=White](generated_hnode26)(10,18){$1$}
\node[Nw=4,Nh=4,linecolor=White](generated_hnode27)(14,18){$1$}
\node[Nw=4,Nh=4,linecolor=White](generated_hnode30)(10,14){$1$}
\node[Nw=4,Nh=4,linecolor=White](generated_hnode31)(14,14){$1$}
\node[Nw=4,Nh=4,linecolor=White](generated_hnode33)(10,10){$1$}
\node[Nw=4,Nh=4,linecolor=White](generated_hnode34)(14,10){$1$}
\node[Nw=4,Nh=4,linecolor=White](generated_hnode36)(14,6){$1\!1$}
\node[Nw=4,Nh=4,linecolor=White](generated_hnode41)(10,2){$1$}
\node[Nw=4,Nh=4,linecolor=White](generated_hnode42)(14,2){$1$}
\node[Nw=4,Nh=4,linecolor=White](generated_vnode0)(18,54){$1$}
\node[Nw=4,Nh=4,linecolor=White](generated_vnode1)(18,50){$1$}
\node[Nw=4,Nh=4,linecolor=White](generated_vnode2)(18,46){$1$}
\node[Nw=4,Nh=4,linecolor=White](generated_vnode4)(22,46){$1\!1$}
\node[Nw=4,Nh=4,linecolor=White](generated_vnode8)(26,54){$1$}
\node[Nw=4,Nh=4,linecolor=White](generated_vnode9)(26,50){$2$}
\node[Nw=4,Nh=4,linecolor=White](generated_vnode10)(26,46){$1$}
\node[Nw=4,Nh=4,linecolor=White](generated_vnode12)(30,54){$1$}
\node[Nw=4,Nh=4,linecolor=White](generated_vnode13)(30,50){$1$}
\node[Nw=4,Nh=4,linecolor=White](generated_vnode14)(30,46){$1$}
\node[Nw=4,Nh=4,linecolor=White](generated_vnode16)(34,54){$1$}
\node[Nw=4,Nh=4,linecolor=White](generated_vnode17)(34,50){$1$}
\node[Nw=4,Nh=4,linecolor=White](generated_vnode18)(34,46){$1$}
\node[Nw=4,Nh=4,linecolor=White](generated_vnode20)(38,54){$1$}
\node[Nw=4,Nh=4,linecolor=White](generated_vnode21)(38,50){$2$}
\node[Nw=4,Nh=4,linecolor=White](generated_vnode22)(38,46){$1$}
\node[Nw=4,Nh=4,linecolor=White](generated_vnode26)(42,54){$1$}
\node[Nw=4,Nh=4,linecolor=White](generated_vnode26)(42,50){$1$}
\node[Nw=4,Nh=4,linecolor=White](generated_vnode27)(42,46){$1$}
\node[Nw=4,Nh=4,linecolor=White](generated_vnode30)(46,54){$1$}
\node[Nw=4,Nh=4,linecolor=White](generated_vnode30)(46,50){$1$}
\node[Nw=4,Nh=4,linecolor=White](generated_vnode31)(46,46){$1$}
\node[Nw=4,Nh=4,linecolor=White](generated_vnode33)(50,54){$1$}
\node[Nw=4,Nh=4,linecolor=White](generated_vnode33)(50,50){$2$}
\node[Nw=4,Nh=4,linecolor=White](generated_vnode34)(50,46){$1$}
\node[Nw=4,Nh=4,linecolor=White](generated_vnode36)(54,46){$1\!1$}
\node[Nw=4,Nh=4,linecolor=White](generated_vnode41)(58,54){$1$}
\node[Nw=4,Nh=4,linecolor=White](generated_vnode41)(58,50){$1$}
\node[Nw=4,Nh=4,linecolor=White](generated_vnode42)(58,46){$1$}
\node[Nw=4,Nh=4,fillcolor=Black,Nmr=0](generated_cell1)(22,42){}
\node[Nw=4,Nh=4,fillcolor=Black,Nmr=0](generated_cell9)(54,42){}
\node[Nw=4,Nh=4,fillcolor=Black,Nmr=0](generated_cell11)(18,38){}
\node[Nw=4,Nh=4,fillcolor=Black,Nmr=0](generated_cell12)(22,38){}
\node[Nw=4,Nh=4,fillcolor=Black,Nmr=0](generated_cell13)(26,38){}
\node[Nw=4,Nh=4,fillcolor=Black,Nmr=0](generated_cell14)(30,38){}
\node[Nw=4,Nh=4,fillcolor=Black,Nmr=0](generated_cell15)(34,38){}
\node[Nw=4,Nh=4,fillcolor=Black,Nmr=0](generated_cell16)(38,38){}
\node[Nw=4,Nh=4,fillcolor=Black,Nmr=0](generated_cell17)(42,38){}
\node[Nw=4,Nh=4,fillcolor=Black,Nmr=0](generated_cell18)(46,38){}
\node[Nw=4,Nh=4,fillcolor=Black,Nmr=0](generated_cell19)(50,38){}
\node[Nw=4,Nh=4,fillcolor=Black,Nmr=0](generated_cell20)(54,38){}
\node[Nw=4,Nh=4,fillcolor=Black,Nmr=0](generated_cell21)(58,38){}
\node[Nw=4,Nh=4,fillcolor=Black,Nmr=0](generated_cell23)(22,34){}
\node[Nw=4,Nh=4,fillcolor=Black,Nmr=0](generated_cell31)(54,34){}
\node[Nw=4,Nh=4,fillcolor=Black,Nmr=0](generated_cell34)(22,30){}
\node[Nw=4,Nh=4,fillcolor=Black,Nmr=0](generated_cell42)(54,30){}
\node[Nw=4,Nh=4,fillcolor=Black,Nmr=0](generated_cell45)(22,26){}
\node[Nw=4,Nh=4,fillcolor=Black,Nmr=0](generated_cell53)(54,26){}
\node[Nw=4,Nh=4,fillcolor=White,Nmr=0](generated_cellLetter55)(18,22){\footnotesize${\sf a}$}
\node[Nw=4,Nh=4,fillcolor=Black,Nmr=0](generated_cell56)(22,22){}
\node[Nw=4,Nh=4,fillcolor=Black,Nmr=0](generated_cell64)(54,22){}
\node[Nw=4,Nh=4,fillcolor=White,Nmr=0](generated_cellLetter65)(58,22){\footnotesize${\sf c}$}
\node[Nw=4,Nh=4,fillcolor=Black,Nmr=0](generated_cell67)(22,18){}
\node[Nw=4,Nh=4,fillcolor=Black,Nmr=0](generated_cell75)(54,18){}
\node[Nw=4,Nh=4,fillcolor=Black,Nmr=0](generated_cell78)(22,14){}
\node[Nw=4,Nh=4,fillcolor=Black,Nmr=0](generated_cell86)(54,14){}
\node[Nw=4,Nh=4,fillcolor=Black,Nmr=0](generated_cell89)(22,10){}
\node[Nw=4,Nh=4,fillcolor=Black,Nmr=0](generated_cell97)(54,10){}
\node[Nw=4,Nh=4,fillcolor=Black,Nmr=0](generated_cell99)(18,6){}
\node[Nw=4,Nh=4,fillcolor=Black,Nmr=0](generated_cell100)(22,6){}
\node[Nw=4,Nh=4,fillcolor=Black,Nmr=0](generated_cell101)(26,6){}
\node[Nw=4,Nh=4,fillcolor=Black,Nmr=0](generated_cell102)(30,6){}
\node[Nw=4,Nh=4,fillcolor=Black,Nmr=0](generated_cell103)(34,6){}
\node[Nw=4,Nh=4,fillcolor=Black,Nmr=0](generated_cell104)(38,6){}
\node[Nw=4,Nh=4,fillcolor=Black,Nmr=0](generated_cell105)(42,6){}
\node[Nw=4,Nh=4,fillcolor=Black,Nmr=0](generated_cell106)(46,6){}
\node[Nw=4,Nh=4,fillcolor=Black,Nmr=0](generated_cell107)(50,6){}
\node[Nw=4,Nh=4,fillcolor=Black,Nmr=0](generated_cell108)(54,6){}
\node[Nw=4,Nh=4,fillcolor=Black,Nmr=0](generated_cell109)(58,6){}
\node[Nw=4,Nh=4,fillcolor=Black,Nmr=0](generated_cell111)(22,2){}
\node[Nw=4,Nh=4,fillcolor=Black,Nmr=0](generated_cell119)(54,2){}
\multiput(16, 0)(4,0){12}{\color{DarkGray} \line(0,1){60}}
\multiput(0, 44)(0,-4){12}{\color{DarkGray} \line(1,0){60}}
\end{picture}
    \label{fig:NonogramSTRAIGHT}}
  \caption[Nonogram gadgets]{Nonogram gadgets. }
  \label{fig:NonogramGadgets}
\end{figure}

The three cells marked with ${\sf a}$, ${\sf b}$ and ${\sf c}$ correspond with the edges; for
the (initially undirected) AND gadget ${\sf a}$ and ${\sf b}$ correspond with the ``red'' ones.
Given the values for the cells corresponding to the edges, the gadgets are within the so-called ``simple'' Nonogram class, and can be easily solved. For a definition and solving algorithm of the ``simple'' nonograms, the reader is referred to~\citeaby{Batenburg2012}.
In any case, the solutions can be easily verified.
They are uniquely characterized
by $\sf{ (a,b,c) } \in \{(0,0,1),(1,0,1),(0,1,1),(1,1,0),(1,1,1)\}$ (AND gadget) and
$\sf{ (a,b,c) } \in \{(1,0,0),(0,1,0),(0,0,1),(1,0,1),(0,1,1),(1,1,0),(1,1,1)\}$ (OR gadget),
implying that the gadgets indeed perform as desired.
If $\sf{ (a,b,c) } = (1,0,1)$, the OR gadget has two solutions,
making use of a so-called switching component.

Now we have:
\begin{theorem}
\textsc{Nonogram} is NP-complete.
\end{theorem}

\begin{proof}
  We can simulate a planar constraint graph with only AND and OR nodes
  on a Nonogram grid using the global layout of Figure~\ref{fig:NonogramGlobalLayout} and the two top gadgets shown in Figure~\ref{fig:NonogramGadgets}. 
The two bottom gadgets from Figure~\ref{fig:NonogramGadgets} can be used for wires, in a straight line or as a corner. 
  Now the arrows can be inserted in a legal way if and only if the resulting Nonogram can be solved.
  
\citeaby{Battista1994} have proven that a graph with maximal degree 3 can always be stored in a square grid of width $v$, where $v$ is the number of vertices contained by the graph. This ensures that the number of rows and columns used in our reduction is bounded linearly by the number of vertices in the corresponding NCL graph. 



\textsc{Nonogram} is clearly in NP, as any potential solution can be verified in polynomial time.\hfill$\Box$
\end{proof}

\section{Dou Shou Qi}\label{sec:doushouqi}

\emph{Dou Shou Qi} (meaning: ``Game of Fighting
Animals''), as described by~\citeaby{Pritchard2007}, is a Chinese board game. In the Western world it is often
called Jungle, The Jungle Game, Jungle Chess, or Animal Chess.
Dou Shou Qi is a two-player abstract strategy game and it contains
some elements from Chess and Stratego as well as some other chess-like
Chinese games (e.g., Banqi). Its origins are not entirely clear, but
it seems that it evolved rather recently (around the 1900s).
It has been suggested by some that the game often ends in a draw,
but preliminary results from~\citeaby{Rijn2013} show a remarkably low
percentage of draws.

The game Dou Shou Qi is not extensively studied in literature.
In~\citeaby{Burnett2010}, a definition of the game is given and an attempt is made
to characterize certain local properties of subproblems that occur
when analyzing the game. These so-called loosely coupled subproblems
can be analyzed separately in contrast to analyzing the problem as a
whole, resulting in a possible speed-up in the overall analysis. A
first complexity result has been obtained by~\citeaby{Rijn2013}.
Dou Shou Qi is proven PSPACE-hard by reduction from logic circuits.
Here, we will prove Dou Shou Qi to be PSPACE-hard by reduction from
Planar Bounded 2CL based on the reduction given by~\citeaby{Rijn2012}.

We cannot prove Dou Shou Qi to be PSPACE-complete; as an unbounded
two-player game it is probably not in PSPACE. 
Dou Shou Qi is clearly in EXPTIME --- like Chess.

\subsection{Definition}

Dou Shou Qi is played on a rectangular board consisting of
$9\times 7$~squares, see Figure~\ref{fig:doushouqi_board}. There are
several different kinds of squares. The \emph{dens} (D), one for each
player, are located in the center of the first and last row and are
protected on all sides by \emph{traps} (T). Furthermore, there are two
bodies of water (W), while the remaining squares are ordinary land
squares.

Each players has eight different pieces representing different
animals with a respective \emph{strength}, according to which they can
\emph{capture} some of the opponent's pieces. Pieces can only capture
a piece of equal or lower strength, with the exception of the weakest
piece which is able to capture the strongest piece. The strength of
the animals from weak to strong is: 1~rat, 2~cat, 3~wolf, 4~dog,
5~panther, 6~tiger, 7~lion, 8~elephant. The initial placement of the
pieces is fixed, see Figure~\ref{fig:doushouqi_board}. Players
alternate turns with white moving first. Each turn a piece must be
moved either one square horizontally or vertically. Pieces are
forbidden to enter their own den and are usually blocked by water. The
rat is the only piece able to move through the water where it is also
capable of capturing, i.e., the enemy rat, but it is forbidden to
capture an elephant while attacking from the water. Lions and tigers
are able to leap over water either horizontally or vertically, but
they are blocked by any rat on the intermediate water squares.

\begin{figure}[!ht]
\begin{center}
\begin{picture}(56, 72)
\linethickness{0.075mm}
\multiput(0, 0)(8,0){8}{\line(0, 1){72}}
\multiput(0, 0)(0,8){10}{\line(1, 0){56}}
\put(4,68){\node[Nw=6,Nh=6,fillcolor=Black,linecolor=Black](00)(0,0){\textcolor{White}{\textbf{7}}} }
\put(20,68){\node[Nw=8,Nh=8,Nmr=0,fillcolor=Black,linecolor=DarkGray](02)(0,0){\textcolor{White}{\textbf{T}}} }
\put(28,68){\node[Nw=8,Nh=8,Nmr=0,fillcolor=Black,linecolor=DarkGray](03)(0,0){\textcolor{White}{\textbf{D}}} }
\put(36,68){\node[Nw=8,Nh=8,Nmr=0,fillcolor=Black,linecolor=DarkGray](04)(0,0){\textcolor{White}{\textbf{T}}} }
\put(52,68){\node[Nw=6,Nh=6,fillcolor=Black,linecolor=Black](06)(0,0){\textcolor{White}{\textbf{6}}} }
\put(12,60){\node[Nw=6,Nh=6,fillcolor=Black,linecolor=Black](91)(0,0){\textcolor{White}{\textbf{4}}} }
\put(28,60){\node[Nw=8,Nh=8,Nmr=0,fillcolor=Black,linecolor=DarkGray](93)(0,0){\textcolor{White}{\textbf{T}}} }
\put(44,60){\node[Nw=6,Nh=6,fillcolor=Black,linecolor=Black](95)(0,0){\textcolor{White}{\textbf{2}}} }
\put(4,52){\node[Nw=6,Nh=6,fillcolor=Black,linecolor=Black](180)(0,0){\textcolor{White}{\textbf{1}}} }
\put(20,52){\node[Nw=6,Nh=6,fillcolor=Black,linecolor=Black](182)(0,0){\textcolor{White}{\textbf{5}}} }
\put(36,52){\node[Nw=6,Nh=6,fillcolor=Black,linecolor=Black](184)(0,0){\textcolor{White}{\textbf{3}}} }
\put(52,52){\node[Nw=6,Nh=6,fillcolor=Black,linecolor=Black](186)(0,0){\textcolor{White}{\textbf{8}}} }
\put(12,44){\node[Nw=6,Nh=6,fillcolor=White,linecolor=White](271)(0,0){{\textbf{W}}} }
\put(20,44){\node[Nw=6,Nh=6,fillcolor=White,linecolor=White](272)(0,0){{\textbf{W}}} }
\put(36,44){\node[Nw=6,Nh=6,fillcolor=White,linecolor=White](274)(0,0){{\textbf{W}}} }
\put(44,44){\node[Nw=6,Nh=6,fillcolor=White,linecolor=White](275)(0,0){{\textbf{W}}} }
\put(12,36){\node[Nw=6,Nh=6,fillcolor=White,linecolor=White](361)(0,0){{\textbf{W}}} }
\put(20,36){\node[Nw=6,Nh=6,fillcolor=White,linecolor=White](362)(0,0){{\textbf{W}}} }
\put(36,36){\node[Nw=6,Nh=6,fillcolor=White,linecolor=White](364)(0,0){{\textbf{W}}} }
\put(44,36){\node[Nw=6,Nh=6,fillcolor=White,linecolor=White](365)(0,0){{\textbf{W}}} }
\put(12,28){\node[Nw=6,Nh=6,fillcolor=White,linecolor=White](451)(0,0){{\textbf{W}}} }
\put(20,28){\node[Nw=6,Nh=6,fillcolor=White,linecolor=White](452)(0,0){{\textbf{W}}} }
\put(36,28){\node[Nw=6,Nh=6,fillcolor=White,linecolor=White](454)(0,0){{\textbf{W}}} }
\put(44,28){\node[Nw=6,Nh=6,fillcolor=White,linecolor=White](455)(0,0){{\textbf{W}}} }
\put(4,20){\node[Nw=6,Nh=6,fillcolor=White,linecolor=Black](540)(0,0){{\textbf{8}}} }
\put(20,20){\node[Nw=6,Nh=6,fillcolor=White,linecolor=Black](542)(0,0){{\textbf{3}}} }
\put(36,20){\node[Nw=6,Nh=6,fillcolor=White,linecolor=Black](544)(0,0){{\textbf{5}}} }
\put(52,20){\node[Nw=6,Nh=6,fillcolor=White,linecolor=Black](546)(0,0){{\textbf{1}}} }
\put(12,12){\node[Nw=6,Nh=6,fillcolor=White,linecolor=Black](631)(0,0){{\textbf{2}}} }
\put(28,12){\node[Nw=6,Nh=6,fillcolor=White,linecolor=White](633)(0,0){{\textbf{T}}} }
\put(44,12){\node[Nw=6,Nh=6,fillcolor=White,linecolor=Black](635)(0,0){{\textbf{4}}} }
\put(4,4){\node[Nw=6,Nh=6,fillcolor=White,linecolor=Black](720)(0,0){{\textbf{6}}} }
\put(20,4){\node[Nw=6,Nh=6,fillcolor=White,linecolor=White](722)(0,0){{\textbf{T}}} }
\put(28,4){\node[Nw=6,Nh=6,fillcolor=White,linecolor=White](723)(0,0){{\textbf{D}}} }
\put(36,4){\node[Nw=6,Nh=6,fillcolor=White,linecolor=White](724)(0,0){{\textbf{T}}} }
\put(52,4){\node[Nw=6,Nh=6,fillcolor=White,linecolor=Black](726)(0,0){{\textbf{7}}} }
\end{picture}
\caption[Dou Shou Qi game board]{A schematic Dou Shou Qi game board
showing the initial position.}
\label{fig:doushouqi_board}
\end{center}
\end{figure}

Pieces can be trapped by the traps surrounding the opponent's den:
their strength is effectively reduced to zero, meaning that they can
be captured by any enemy piece. The objective of the game is to place
a piece in the opponents den or to eliminate all of the opponent's
pieces. A stalemate position is declared a draw.

\subsection{PSPACE-hardness}

The Bounded 2CL graph will be simulated on a $m \times n$~board, where
both players have $k$ pieces. Whether a natural generalization of the
game would imply that the $k$ pieces all have a strength from the
interval~$[1,8]$ or a distinct strength from the interval~$[1,k]$ is
open for discussion. In our reduction all pieces have a strength from
the interval~$[2,5]$, excluding all pieces with special capabilities.
The original game board contains several properties, i.e., clustered
water squares, narrow paths between the water, traps surrounding the
dens, which are symmetrical and highly regular. Which of these
properties should be preserved on a generalized game board is open for
debate, however in our reduction we took the liberty of using water
squares and traps freely in the gadgets.

We will show a complexity proof for the decision problem 
\textsc{Dou Shou Qi}: given a Dou Shou Qi position, does the player on 
turn have a forced win?


\begin{figure}[!ht]
\begin{center}
\subfigure[AND gadget]{
\begin{picture}(25, 45)
\linethickness{0.075mm}
\multiput(0, 0)(5,0){6}{\line(0, 1){45}}
\multiput(0, 0)(0,5){10}{\line(1, 0){25}}
\put(2.5,42.5){\node[Nw=3,Nh=3,fillcolor=White,linecolor=White](00)(0,0){{\textbf{W}}} }
\put(7.5,42.5){\node[Nw=3,Nh=3,fillcolor=White,linecolor=White](01)(0,0){{\textbf{W}}} }
\put(17.5,42.5){\node[Nw=3,Nh=3,fillcolor=White,linecolor=White](03)(0,0){{\textbf{W}}} }
\put(22.5,42.5){\node[Nw=3,Nh=3,fillcolor=White,linecolor=White](04)(0,0){{\textbf{W}}} }
\put(2.5,37.5){\node[Nw=3,Nh=3,fillcolor=White,linecolor=White](90)(0,0){{\textbf{W}}} }
\put(12.5,37.5){\node[Nw=3,Nh=3,fillcolor=White,linecolor=White](90)(0,0){{\textbf{T}}} }
\put(22.5,37.5){\node[Nw=3,Nh=3,fillcolor=White,linecolor=White](94)(0,0){{\textbf{W}}} }
\put(2.5,32.5){\node[Nw=3,Nh=3,fillcolor=White,linecolor=White](180)(0,0){{\textbf{W}}} }
\put(7.5,32.5){\node[Nw=3,Nh=3,fillcolor=White,linecolor=White](90)(0,0){{\textbf{T}}} }
\put(12.5,32.5){\node[Nw=3,Nh=3,fillcolor=White,linecolor=Black](182)(0,0){{\textbf{2}}} }
\put(17.5,32.5){\node[Nw=3,Nh=3,fillcolor=White,linecolor=White](90)(0,0){{\textbf{T}}} }
\put(22.5,32.5){\node[Nw=3,Nh=3,fillcolor=White,linecolor=White](184)(0,0){{\textbf{W}}} }
\put(2.5,27.5){\node[Nw=3,Nh=3,fillcolor=White,linecolor=White](270)(0,0){{\textbf{W}}} }
\put(12.5,27.5){\node[Nw=3,Nh=3,fillcolor=White,linecolor=White](90)(0,0){{\textbf{T}}} }
\put(22.5,27.5){\node[Nw=3,Nh=3,fillcolor=White,linecolor=White](274)(0,0){{\textbf{W}}} }
\put(2.5,22.5){\node[Nw=3,Nh=3,fillcolor=White,linecolor=White](360)(0,0){{\textbf{W}}} }
\put(12.5,22.5){\node[Nw=3,Nh=3,fillcolor=Black,linecolor=Black](362)(0,0){\textcolor{White}{\textbf{5}}} }
\put(22.5,22.5){\node[Nw=3,Nh=3,fillcolor=White,linecolor=White](364)(0,0){{\textbf{W}}} }
\put(2.5,17.5){\node[Nw=3,Nh=3,fillcolor=White,linecolor=White](450)(0,0){{\textbf{W}}} }
\put(12.5,17.5){\node[Nw=3,Nh=3,fillcolor=White,linecolor=White](90)(0,0){{\textbf{T}}} }
\put(22.5,17.5){\node[Nw=3,Nh=3,fillcolor=White,linecolor=White](454)(0,0){{\textbf{W}}} }
\put(2.5,12.5){\node[Nw=3,Nh=3,fillcolor=White,linecolor=White](540)(0,0){{\textbf{W}}} }
\put(7.5,12.5){\node[Nw=3,Nh=3,fillcolor=White,linecolor=White](90)(0,0){{\textbf{T}}} }
\put(12.5,12.5){\node[Nw=3,Nh=3,fillcolor=White,linecolor=Black](542)(0,0){{\textbf{2}}} }
\put(17.5,12.5){\node[Nw=3,Nh=3,fillcolor=White,linecolor=White](90)(0,0){{\textbf{T}}} }
\put(22.5,12.5){\node[Nw=3,Nh=3,fillcolor=White,linecolor=White](544)(0,0){{\textbf{W}}} }
\put(2.5,7.5){\node[Nw=3,Nh=3,fillcolor=White,linecolor=White](630)(0,0){{\textbf{W}}} }
\put(12.5,7.5){\node[Nw=3,Nh=3,fillcolor=White,linecolor=White](90)(0,0){{\textbf{T}}} }
\put(22.5,7.5){\node[Nw=3,Nh=3,fillcolor=White,linecolor=White](634)(0,0){{\textbf{W}}} }
\put(2.5,2.5){\node[Nw=3,Nh=3,fillcolor=White,linecolor=White](720)(0,0){{\textbf{W}}} }
\put(12.5,2.5){\node[Nw=3,Nh=3,fillcolor=White,linecolor=White](722)(0,0){{\textbf{W}}} }
\put(22.5,2.5){\node[Nw=3,Nh=3,fillcolor=White,linecolor=White](724)(0,0){{\textbf{W}}} }
\end{picture}
\label{fig:doushouqi_and}}
\subfigure[OR gadget]{
\begin{picture}(25, 45)
\linethickness{0.075mm}
\multiput(0, 0)(5,0){6}{\line(0, 1){45}}
\multiput(0, 0)(0,5){10}{\line(1, 0){25}}
\put(2.5,42.5){\node[Nw=3,Nh=3,fillcolor=White,linecolor=White](00)(0,0){{\textbf{W}}} }
\put(7.5,42.5){\node[Nw=3,Nh=3,fillcolor=White,linecolor=White](01)(0,0){{\textbf{W}}} }
\put(17.5,42.5){\node[Nw=3,Nh=3,fillcolor=White,linecolor=White](03)(0,0){{\textbf{W}}} }
\put(22.5,42.5){\node[Nw=3,Nh=3,fillcolor=White,linecolor=White](04)(0,0){{\textbf{W}}} }
\put(2.5,37.5){\node[Nw=3,Nh=3,fillcolor=White,linecolor=White](90)(0,0){{\textbf{W}}} }
\put(12.5,37.5){\node[Nw=3,Nh=3,fillcolor=White,linecolor=White](90)(0,0){{\textbf{T}}} }
\put(22.5,37.5){\node[Nw=3,Nh=3,fillcolor=White,linecolor=White](94)(0,0){{\textbf{W}}} }
\put(2.5,32.5){\node[Nw=3,Nh=3,fillcolor=White,linecolor=White](180)(0,0){{\textbf{W}}} }
\put(7.5,32.5){\node[Nw=3,Nh=3,fillcolor=White,linecolor=White](90)(0,0){{\textbf{T}}} }
\put(12.5,32.5){\node[Nw=3,Nh=3,fillcolor=White,linecolor=Black](182)(0,0){{\textbf{2}}} }
\put(17.5,32.5){\node[Nw=3,Nh=3,fillcolor=White,linecolor=White](90)(0,0){{\textbf{T}}} }
\put(22.5,32.5){\node[Nw=3,Nh=3,fillcolor=White,linecolor=White](184)(0,0){{\textbf{W}}} }
\put(2.5,27.5){\node[Nw=3,Nh=3,fillcolor=White,linecolor=White](270)(0,0){{\textbf{W}}} }
\put(12.5,27.5){\node[Nw=3,Nh=3,fillcolor=White,linecolor=White](90)(0,0){{\textbf{T}}} }
\put(22.5,27.5){\node[Nw=3,Nh=3,fillcolor=White,linecolor=White](274)(0,0){{\textbf{W}}} }
\put(2.5,22.5){\node[Nw=3,Nh=3,fillcolor=White,linecolor=White](360)(0,0){{\textbf{W}}} }
\put(12.5,22.5){\node[Nw=3,Nh=3,fillcolor=Black,linecolor=Black](362)(0,0){\textcolor{White}{\textbf{5}}} }
\put(22.5,22.5){\node[Nw=3,Nh=3,fillcolor=White,linecolor=White](364)(0,0){{\textbf{W}}} }
\put(2.5,17.5){\node[Nw=3,Nh=3,fillcolor=White,linecolor=White](450)(0,0){{\textbf{W}}} }
\put(12.5,17.5){\node[Nw=3,Nh=3,fillcolor=White,linecolor=White](90)(0,0){{\textbf{T}}} }
\put(22.5,17.5){\node[Nw=3,Nh=3,fillcolor=White,linecolor=White](454)(0,0){{\textbf{W}}} }
\put(2.5,12.5){\node[Nw=3,Nh=3,fillcolor=White,linecolor=White](540)(0,0){{\textbf{W}}} }
\put(7.5,12.5){\node[Nw=3,Nh=3,fillcolor=White,linecolor=White](90)(0,0){{\textbf{T}}} }
\put(12.5,12.5){\node[Nw=3,Nh=3,fillcolor=White,linecolor=Black](542)(0,0){{\textbf{2}}} }
\put(17.5,12.5){\node[Nw=3,Nh=3,fillcolor=White,linecolor=White](90)(0,0){{\textbf{T}}} }
\put(22.5,12.5){\node[Nw=3,Nh=3,fillcolor=White,linecolor=White](544)(0,0){{\textbf{W}}} }
\put(2.5,7.5){\node[Nw=3,Nh=3,fillcolor=White,linecolor=White](630)(0,0){{\textbf{W}}} }
\put(12.5,7.5){\node[Nw=3,Nh=3,fillcolor=White,linecolor=White](90)(0,0){{\textbf{T}}} }
\put(17.5,7.5){\node[Nw=3,Nh=3,fillcolor=White,linecolor=Black](633)(0,0){{\textbf{4}}} }
\put(22.5,7.5){\node[Nw=3,Nh=3,fillcolor=White,linecolor=White](634)(0,0){{\textbf{W}}} }
\put(2.5,2.5){\node[Nw=3,Nh=3,fillcolor=White,linecolor=White](720)(0,0){{\textbf{W}}} }
\put(12.5,2.5){\node[Nw=3,Nh=3,fillcolor=White,linecolor=White](722)(0,0){{\textbf{W}}} }
\put(22.5,2.5){\node[Nw=3,Nh=3,fillcolor=White,linecolor=White](724)(0,0){{\textbf{W}}} }
\end{picture}
\label{fig:doushouqi_or}}
\subfigure[FANOUT gadget]{
\begin{picture}(30, 25)
\linethickness{0.075mm}
\multiput(0, 0)(5,0){7}{\line(0, 1){25}}
\multiput(0, 0)(0,5){6}{\line(1, 0){30}}
\put(2.5,22.5){\node[Nw=3,Nh=3,fillcolor=White,linecolor=White](00)(0,0){{\textbf{W}}} }
\put(12.5,22.5){\node[Nw=3,Nh=3,fillcolor=White,linecolor=White](02)(0,0){{\textbf{W}}} }
\put(17.5,22.5){\node[Nw=3,Nh=3,fillcolor=White,linecolor=White](03)(0,0){{\textbf{W}}} }
\put(27.5,22.5){\node[Nw=3,Nh=3,fillcolor=White,linecolor=White](05)(0,0){{\textbf{W}}} }
\put(2.5,17.5){\node[Nw=3,Nh=3,fillcolor=White,linecolor=White](50)(0,0){{\textbf{W}}} }
\put(27.5,17.5){\node[Nw=3,Nh=3,fillcolor=White,linecolor=White](55)(0,0){{\textbf{W}}} }
\put(2.5,12.5){\node[Nw=3,Nh=3,fillcolor=White,linecolor=White](100)(0,0){{\textbf{W}}} }
\put(7.5,12.5){\node[Nw=5,Nh=5,Nmr=0,fillcolor=Black,linecolor=Black](101)(0,0){\textcolor{White}{\textbf{T}}} }
\put(12.5,12.5){\node[Nw=3,Nh=3,fillcolor=Black,linecolor=Black](102)(0,0){\textcolor{White}{\textbf{3}}} }
\put(27.5,12.5){\node[Nw=3,Nh=3,fillcolor=White,linecolor=White](105)(0,0){{\textbf{W}}} }
\put(2.5,7.5){\node[Nw=3,Nh=3,fillcolor=White,linecolor=White](150)(0,0){{\textbf{W}}} }
\put(7.5,7.5){\node[Nw=3,Nh=3,fillcolor=White,linecolor=Black](151)(0,0){{\textbf{4}}} }
\put(12.5,7.5){\node[Nw=5,Nh=5,Nmr=0,fillcolor=Black,linecolor=Black](101)(0,0){\textcolor{White}{\textbf{T}}} }
\put(27.5,7.5){\node[Nw=3,Nh=3,fillcolor=White,linecolor=White](155)(0,0){{\textbf{W}}} }
\put(2.5,2.5){\node[Nw=3,Nh=3,fillcolor=White,linecolor=White](200)(0,0){{\textbf{W}}} }
\put(7.5,2.5){\node[Nw=3,Nh=3,fillcolor=White,linecolor=White](201)(0,0){{\textbf{W}}} }
\put(12.5,2.5){\node[Nw=3,Nh=3,fillcolor=White,linecolor=White](202)(0,0){{\textbf{W}}} }
\put(22.5,2.5){\node[Nw=3,Nh=3,fillcolor=White,linecolor=White](204)(0,0){{\textbf{W}}} }
\put(27.5,2.5){\node[Nw=3,Nh=3,fillcolor=White,linecolor=White](205)(0,0){{\textbf{W}}} }
\end{picture}
\label{fig:doushouqi_fanout}}
\subfigure[CHOICE gadget]{
\begin{picture}(25, 15)
\linethickness{0.075mm}
\multiput(0, 0)(5,0){6}{\line(0, 1){15}}
\multiput(0, 0)(0,5){4}{\line(1, 0){25}}
\put(2.5,12.5){\node[Nw=3,Nh=3,fillcolor=White,linecolor=White](00)(0,0){{\textbf{W}}} }
\put(12.5,12.5){\node[Nw=3,Nh=3,fillcolor=White,linecolor=White](02)(0,0){{\textbf{W}}} }
\put(22.5,12.5){\node[Nw=3,Nh=3,fillcolor=White,linecolor=White](04)(0,0){{\textbf{W}}} }
\put(2.5,7.5){\node[Nw=3,Nh=3,fillcolor=White,linecolor=White](30)(0,0){{\textbf{W}}} }
\put(22.5,7.5){\node[Nw=3,Nh=3,fillcolor=White,linecolor=White](34)(0,0){{\textbf{W}}} }
\put(2.5,2.5){\node[Nw=3,Nh=3,fillcolor=White,linecolor=White](60)(0,0){{\textbf{W}}} }
\put(7.5,2.5){\node[Nw=3,Nh=3,fillcolor=White,linecolor=White](61)(0,0){{\textbf{W}}} }
\put(17.5,2.5){\node[Nw=3,Nh=3,fillcolor=White,linecolor=White](63)(0,0){{\textbf{W}}} }
\put(22.5,2.5){\node[Nw=3,Nh=3,fillcolor=White,linecolor=White](64)(0,0){{\textbf{W}}} }
\end{picture}
\label{fig:doushouqi_coice}}
\subfigure[VARIABLE gadget]{
\begin{picture}(25, 25)
\linethickness{0.075mm}
\multiput(0, 0)(5,0){6}{\line(0, 1){25}}
\multiput(0, 0)(0,5){6}{\line(1, 0){25}}
\put(2.5,22.5){\node[Nw=3,Nh=3,fillcolor=White,linecolor=White](00)(0,0){{\textbf{W}}} }
\put(7.5,22.5){\node[Nw=3,Nh=3,fillcolor=White,linecolor=White](01)(0,0){{\textbf{W}}} }
\put(17.5,22.5){\node[Nw=3,Nh=3,fillcolor=White,linecolor=White](03)(0,0){{\textbf{W}}} }
\put(22.5,22.5){\node[Nw=3,Nh=3,fillcolor=White,linecolor=White](04)(0,0){{\textbf{W}}} }
\put(2.5,17.5){\node[Nw=3,Nh=3,fillcolor=White,linecolor=White](50)(0,0){{\textbf{W}}} }
\put(22.5,17.5){\node[Nw=3,Nh=3,fillcolor=White,linecolor=White](54)(0,0){{\textbf{W}}} }
\put(2.5,12.5){\node[Nw=3,Nh=3,fillcolor=White,linecolor=White](100)(0,0){{\textbf{W}}} }
\put(12.5,12.5){\node[Nw=3,Nh=3,fillcolor=Black,linecolor=Black](102)(0,0){\textcolor{White}{\textbf{4}}} }
\put(22.5,12.5){\node[Nw=3,Nh=3,fillcolor=White,linecolor=White](104)(0,0){{\textbf{W}}} }
\put(2.5,7.5){\node[Nw=3,Nh=3,fillcolor=White,linecolor=White](150)(0,0){{\textbf{W}}} }
\put(12.5,7.5){\node[Nw=3,Nh=3,fillcolor=White,linecolor=Black](152)(0,0){{\textbf{4}}} }
\put(22.5,7.5){\node[Nw=3,Nh=3,fillcolor=White,linecolor=White](154)(0,0){{\textbf{W}}} }
\put(2.5,2.5){\node[Nw=3,Nh=3,fillcolor=White,linecolor=White](200)(0,0){{\textbf{W}}} }
\put(7.5,2.5){\node[Nw=3,Nh=3,fillcolor=White,linecolor=White](201)(0,0){{\textbf{W}}} }
\put(17.5,2.5){\node[Nw=3,Nh=3,fillcolor=White,linecolor=White](203)(0,0){{\textbf{W}}} }
\put(22.5,2.5){\node[Nw=3,Nh=3,fillcolor=White,linecolor=White](204)(0,0){{\textbf{W}}} }
\end{picture}
\label{fig:doushouqi_variable}}
\caption{Dou Shou Qi gadgets.}
\label{fig:doushouqi_gadgets}
\end{center}
\end{figure}
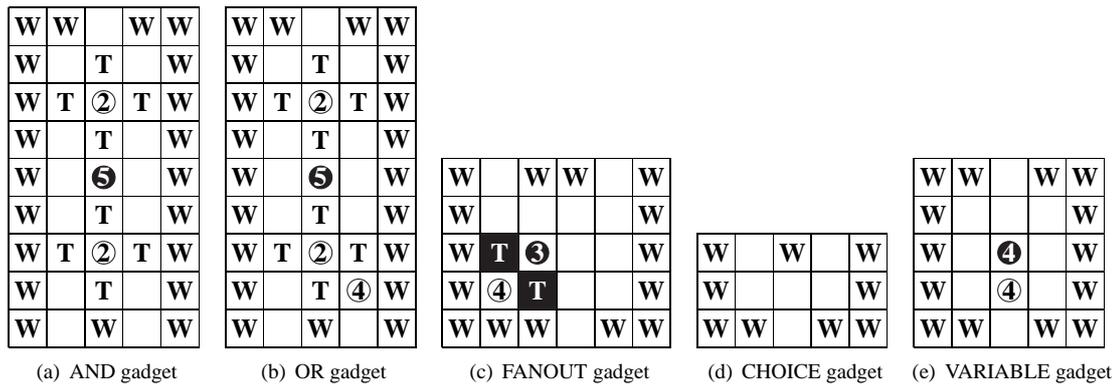

We will reduce from Planar Bounded 2CL. The main gadgets are shown in 
Figure~\ref{fig:doushouqi_gadgets}.  
The reversal of an edge in the original Bounded 2CL graph will be
modeled as the movement of a white dog (strength~4) into another
gadget. The VARIABLE vertex (in its initial state) can  
be reversed by the current player. The same is true for
the VARIABLE gadget. Since both pieces are of the same
strength, the current player can capture the piece of the 
other player, and move its piece through to the next gadget. 

There are some additional issues that need to be addressed. 
First, white dogs that enter a gadget should not be allowed 
to go back into the previous gadget. Next, white dogs in 
a FANOUT gadget should not be allowed to move through the same 
exit twice. Finally, black pieces in a VARIABLE gadget should 
not be allowed to leave
the gadget through the exit corresponding to the white edge in 
the graph game.
In order to prevent this behavior, we have created some additional
support gadgets, that will be attached to the inputs and outputs of the 
gadgets. These are shown in Figure~\ref{fig:doushouqi_support}.

The construction shown in Figure~\ref{fig:doushouqi_black} is 
called a \emph{black edge protector}. The white player can move a dog 
from bottom to top, but not the other way around. When a white dog 
enters the construction, the black piece will retreat behind either 
the left trap or the right trap, and the white dog can pass. 
When passed, the black piece moves back to its original position. 
The white piece cannot move back, it would be captured when 
entering the traps. The construction shown in 
Figure~\ref{fig:doushouqi_white} is a \emph{white edge protector}, 
it allows only white pieces to pass. Black pieces can be captured 
upon entering a trap. Note that these constraints do not apply when 
the opposing player attacks from both sides. We will show further 
on how to deal with this. 
The construction in Figure~\ref{fig:doushouqi_outflow} is an outflow 
protector, with the left entrance square as input and the right entrance 
square as output. It ensures that upon arrival of either one or two
white pieces, only one can pass.

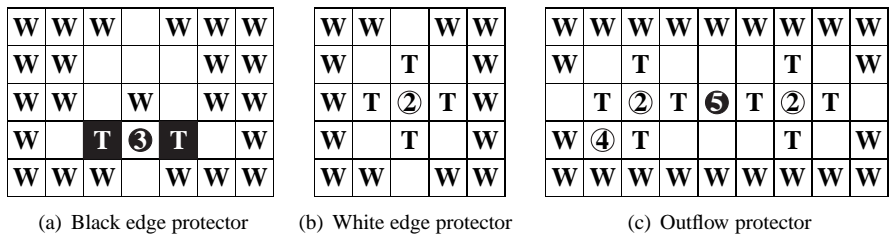
\begin{figure}[!ht]
\begin{center}
\subfigure[Black edge protector]{
\begin{picture}(35, 25)
\linethickness{0.075mm}
\multiput(0, 0)(5,0){8}{\line(0, 1){25}}
\multiput(0, 0)(0,5){6}{\line(1, 0){35}}
\put(2.5,22.5){\node[Nw=3,Nh=3,fillcolor=White,linecolor=White](00)(0,0){{\textbf{W}}} }
\put(7.5,22.5){\node[Nw=3,Nh=3,fillcolor=White,linecolor=White](01)(0,0){{\textbf{W}}} }
\put(12.5,22.5){\node[Nw=3,Nh=3,fillcolor=White,linecolor=White](02)(0,0){{\textbf{W}}} }
\put(22.5,22.5){\node[Nw=3,Nh=3,fillcolor=White,linecolor=White](04)(0,0){{\textbf{W}}} }
\put(27.5,22.5){\node[Nw=3,Nh=3,fillcolor=White,linecolor=White](05)(0,0){{\textbf{W}}} }
\put(32.5,22.5){\node[Nw=3,Nh=3,fillcolor=White,linecolor=White](06)(0,0){{\textbf{W}}} }
\put(2.5,17.5){\node[Nw=3,Nh=3,fillcolor=White,linecolor=White](50)(0,0){{\textbf{W}}} }
\put(7.5,17.5){\node[Nw=3,Nh=3,fillcolor=White,linecolor=White](51)(0,0){{\textbf{W}}} }
\put(27.5,17.5){\node[Nw=3,Nh=3,fillcolor=White,linecolor=White](55)(0,0){{\textbf{W}}} }
\put(32.5,17.5){\node[Nw=3,Nh=3,fillcolor=White,linecolor=White](56)(0,0){{\textbf{W}}} }
\put(2.5,12.5){\node[Nw=3,Nh=3,fillcolor=White,linecolor=White](100)(0,0){{\textbf{W}}} }
\put(7.5,12.5){\node[Nw=3,Nh=3,fillcolor=White,linecolor=White](101)(0,0){{\textbf{W}}} }
\put(17.5,12.5){\node[Nw=3,Nh=3,fillcolor=White,linecolor=White](103)(0,0){{\textbf{W}}} }
\put(27.5,12.5){\node[Nw=3,Nh=3,fillcolor=White,linecolor=White](105)(0,0){{\textbf{W}}} }
\put(32.5,12.5){\node[Nw=3,Nh=3,fillcolor=White,linecolor=White](106)(0,0){{\textbf{W}}} }
\put(2.5,7.5){\node[Nw=3,Nh=3,fillcolor=White,linecolor=White](150)(0,0){{\textbf{W}}} }
\put(12.5,7.5){\node[Nw=5,Nh=5,Nmr=0,fillcolor=Black,linecolor=Black](152)(0,0){\textcolor{White}{\textbf{T}}} }
\put(17.5,7.5){\node[Nw=3,Nh=3,fillcolor=Black,linecolor=Black](153)(0,0){\textcolor{White}{\textbf{3}}} }
\put(22.5,7.5){\node[Nw=5,Nh=5,Nmr=0,fillcolor=Black,linecolor=Black](154)(0,0){\textcolor{White}{\textbf{T}}} }
\put(32.5,7.5){\node[Nw=3,Nh=3,fillcolor=White,linecolor=White](156)(0,0){{\textbf{W}}} }
\put(2.5,2.5){\node[Nw=3,Nh=3,fillcolor=White,linecolor=White](200)(0,0){{\textbf{W}}} }
\put(7.5,2.5){\node[Nw=3,Nh=3,fillcolor=White,linecolor=White](201)(0,0){{\textbf{W}}} }
\put(12.5,2.5){\node[Nw=3,Nh=3,fillcolor=White,linecolor=White](202)(0,0){{\textbf{W}}} }
\put(22.5,2.5){\node[Nw=3,Nh=3,fillcolor=White,linecolor=White](204)(0,0){{\textbf{W}}} }
\put(27.5,2.5){\node[Nw=3,Nh=3,fillcolor=White,linecolor=White](205)(0,0){{\textbf{W}}} }
\put(32.5,2.5){\node[Nw=3,Nh=3,fillcolor=White,linecolor=White](206)(0,0){{\textbf{W}}} }
\end{picture}
\label{fig:doushouqi_black}}
\subfigure[White edge protector]{
\hspace{0.1cm}
\begin{picture}(25, 25)
\linethickness{0.075mm}
\multiput(0, 0)(5,0){6}{\line(0, 1){25}}
\multiput(0, 0)(0,5){6}{\line(1, 0){25}}
\put(2.5,22.5){\node[Nw=3,Nh=3,fillcolor=White,linecolor=White](00)(0,0){{\textbf{W}}} }
\put(7.5,22.5){\node[Nw=3,Nh=3,fillcolor=White,linecolor=White](01)(0,0){{\textbf{W}}} }
\put(17.5,22.5){\node[Nw=3,Nh=3,fillcolor=White,linecolor=White](03)(0,0){{\textbf{W}}} }
\put(22.5,22.5){\node[Nw=3,Nh=3,fillcolor=White,linecolor=White](04)(0,0){{\textbf{W}}} }
\put(2.5,17.5){\node[Nw=3,Nh=3,fillcolor=White,linecolor=White](50)(0,0){{\textbf{W}}} }
\put(12.5,17.5){\node[Nw=3,Nh=3,fillcolor=White,linecolor=White](50)(0,0){{\textbf{T}}} }
\put(22.5,17.5){\node[Nw=3,Nh=3,fillcolor=White,linecolor=White](54)(0,0){{\textbf{W}}} }
\put(2.5,12.5){\node[Nw=3,Nh=3,fillcolor=White,linecolor=White](100)(0,0){{\textbf{W}}} }
\put(7.5,12.5){\node[Nw=3,Nh=3,fillcolor=White,linecolor=White](50)(0,0){{\textbf{T}}} }
\put(12.5,12.5){\node[Nw=3,Nh=3,fillcolor=White,linecolor=Black](102)(0,0){{\textbf{2}}} }
\put(17.5,12.5){\node[Nw=3,Nh=3,fillcolor=White,linecolor=White](50)(0,0){{\textbf{T}}} }
\put(22.5,12.5){\node[Nw=3,Nh=3,fillcolor=White,linecolor=White](104)(0,0){{\textbf{W}}} }
\put(2.5,7.5){\node[Nw=3,Nh=3,fillcolor=White,linecolor=White](150)(0,0){{\textbf{W}}} }
\put(12.5,7.5){\node[Nw=3,Nh=3,fillcolor=White,linecolor=White](50)(0,0){{\textbf{T}}} }
\put(22.5,7.5){\node[Nw=3,Nh=3,fillcolor=White,linecolor=White](154)(0,0){{\textbf{W}}} }
\put(2.5,2.5){\node[Nw=3,Nh=3,fillcolor=White,linecolor=White](200)(0,0){{\textbf{W}}} }
\put(7.5,2.5){\node[Nw=3,Nh=3,fillcolor=White,linecolor=White](201)(0,0){{\textbf{W}}} }
\put(17.5,2.5){\node[Nw=3,Nh=3,fillcolor=White,linecolor=White](203)(0,0){{\textbf{W}}} }
\put(22.5,2.5){\node[Nw=3,Nh=3,fillcolor=White,linecolor=White](204)(0,0){{\textbf{W}}} }
\end{picture}
\label{fig:doushouqi_white}}
\hspace{0.1cm}
\subfigure[Outflow protector]{
\begin{picture}(45, 25)
\linethickness{0.075mm}
\multiput(0, 0)(5,0){10}{\line(0, 1){25}}
\multiput(0, 0)(0,5){6}{\line(1, 0){45}}
\put(2.5,22.5){\node[Nw=3,Nh=3,fillcolor=White,linecolor=White](00)(0,0){{\textbf{W}}} }
\put(7.5,22.5){\node[Nw=3,Nh=3,fillcolor=White,linecolor=White](01)(0,0){{\textbf{W}}} }
\put(12.5,22.5){\node[Nw=3,Nh=3,fillcolor=White,linecolor=White](02)(0,0){{\textbf{W}}} }
\put(17.5,22.5){\node[Nw=3,Nh=3,fillcolor=White,linecolor=White](03)(0,0){{\textbf{W}}} }
\put(22.5,22.5){\node[Nw=3,Nh=3,fillcolor=White,linecolor=White](04)(0,0){{\textbf{W}}} }
\put(27.5,22.5){\node[Nw=3,Nh=3,fillcolor=White,linecolor=White](05)(0,0){{\textbf{W}}} }
\put(32.5,22.5){\node[Nw=3,Nh=3,fillcolor=White,linecolor=White](06)(0,0){{\textbf{W}}} }
\put(37.5,22.5){\node[Nw=3,Nh=3,fillcolor=White,linecolor=White](07)(0,0){{\textbf{W}}} }
\put(42.5,22.5){\node[Nw=3,Nh=3,fillcolor=White,linecolor=White](08)(0,0){{\textbf{W}}} }
\put(2.5,17.5){\node[Nw=3,Nh=3,fillcolor=White,linecolor=White](50)(0,0){{\textbf{W}}} }
\put(12.5,17.5){\node[Nw=3,Nh=3,fillcolor=White,linecolor=White](50)(0,0){{\textbf{T}}} }
\put(32.5,17.5){\node[Nw=3,Nh=3,fillcolor=White,linecolor=White](50)(0,0){{\textbf{T}}} }
\put(42.5,17.5){\node[Nw=3,Nh=3,fillcolor=White,linecolor=White](58)(0,0){{\textbf{W}}} }
\put(7.5,12.5){\node[Nw=3,Nh=3,fillcolor=White,linecolor=White](50)(0,0){{\textbf{T}}} }
\put(12.5,12.5){\node[Nw=3,Nh=3,fillcolor=White,linecolor=Black](102)(0,0){{\textbf{2}}} }
\put(17.5,12.5){\node[Nw=3,Nh=3,fillcolor=White,linecolor=White](50)(0,0){{\textbf{T}}} }
\put(22.5,12.5){\node[Nw=3,Nh=3,fillcolor=Black,linecolor=Black](104)(0,0){\textcolor{White}{\textbf{5}}} }
\put(27.5,12.5){\node[Nw=3,Nh=3,fillcolor=White,linecolor=White](50)(0,0){{\textbf{T}}} }
\put(32.5,12.5){\node[Nw=3,Nh=3,fillcolor=White,linecolor=Black](106)(0,0){{\textbf{2}}} }
\put(37.5,12.5){\node[Nw=3,Nh=3,fillcolor=White,linecolor=White](50)(0,0){{\textbf{T}}} }
\put(2.5,7.5){\node[Nw=3,Nh=3,fillcolor=White,linecolor=White](150)(0,0){{\textbf{W}}} }
\put(7.5,7.5){\node[Nw=3,Nh=3,fillcolor=White,linecolor=Black](151)(0,0){{\textbf{4}}} }
\put(12.5,7.5){\node[Nw=3,Nh=3,fillcolor=White,linecolor=White](50)(0,0){{\textbf{T}}} }
\put(32.5,7.5){\node[Nw=3,Nh=3,fillcolor=White,linecolor=White](50)(0,0){{\textbf{T}}} }
\put(42.5,7.5){\node[Nw=3,Nh=3,fillcolor=White,linecolor=White](158)(0,0){{\textbf{W}}} }
\put(2.5,2.5){\node[Nw=3,Nh=3,fillcolor=White,linecolor=White](200)(0,0){{\textbf{W}}} }
\put(7.5,2.5){\node[Nw=3,Nh=3,fillcolor=White,linecolor=White](201)(0,0){{\textbf{W}}} }
\put(12.5,2.5){\node[Nw=3,Nh=3,fillcolor=White,linecolor=White](202)(0,0){{\textbf{W}}} }
\put(17.5,2.5){\node[Nw=3,Nh=3,fillcolor=White,linecolor=White](203)(0,0){{\textbf{W}}} }
\put(22.5,2.5){\node[Nw=3,Nh=3,fillcolor=White,linecolor=White](204)(0,0){{\textbf{W}}} }
\put(27.5,2.5){\node[Nw=3,Nh=3,fillcolor=White,linecolor=White](205)(0,0){{\textbf{W}}} }
\put(32.5,2.5){\node[Nw=3,Nh=3,fillcolor=White,linecolor=White](206)(0,0){{\textbf{W}}} }
\put(37.5,2.5){\node[Nw=3,Nh=3,fillcolor=White,linecolor=White](207)(0,0){{\textbf{W}}} }
\put(42.5,2.5){\node[Nw=3,Nh=3,fillcolor=White,linecolor=White](208)(0,0){{\textbf{W}}} }
\end{picture}
\label{fig:doushouqi_outflow}}
\caption[Dou Shou Qi support constructions]{Constructions used to support gadgets.}
\label{fig:doushouqi_support}
\end{center}
\end{figure}

A chain of two white edge protectors, one black edge protector 
and another two white edge protectors is called a 
\emph{one-way channel}. 
First, it ensures that no black piece can move through it. It 
will be captured upon entering a white edge protector. After 
a capture, the white cat can resume its position, preventing black 
pieces from passing, regardless of their number. Because it is always 
adjacent to another white edge protector, even an attack from both 
sides is useless. Next, it ensures that all black pieces within 
are unable to move out of the construction they started in, by 
the same argument. Finally, linking several one-way channels to 
each other ensures that white pieces can move through it in only 
one direction. White pieces that move in the opposite direction 
will be stopped at the black edge protector. Indeed, when having 
a piece at both the input and the output, the white player can 
enable its piece at the output to move back through this gadget. 
However, 
in order to pass a number of subsequent black edge protectors, the 
white player needs an equal number of white dogs at the input to 
ensure such a passing. There can never be more than two white pieces 
at the input of a one-way channel, thus linking three one-way 
channels together prevents white pieces from moving into the forbidden 
direction. 
A gadget protector is a chain of three one-way channels, one outflow 
protector and another three one-way channels. The gadget protector 
is attached to every entrance of the gadgets shown in 
Figure~\ref{fig:doushouqi_gadgets}, ensuring that these 
facilitate exactly the same behavior as their equivalents in the 
graph game. 

Now we have:
\begin{theorem}
\textsc{Dou Shou Qi} is PSPACE-hard.
\end{theorem}

\begin{proof}
Reduction from Bounded 2CL. Given a planar constraint graph made of AND, OR,
FANOUT, CHOICE and VARIABLE vertices, we construct a corresponding Dou
Shou Qi game board where the white player has a forced win if and only
if (s)he has a forced win on the original Bounded 2CL graph; otherwise
the black player has a forced win. Note that there are no draws in
Bounded 2CL, neither are there in the reduction by optimal play.

The target edge will be represented by a gadget containing a black den,
and it will have a black edge protector 
(Figure~\ref{fig:doushouqi_black}) in front of it, preventing other
pieces than the white dogs from entering it. The white player can move a 
piece into this gadget if and
only if (s)he can set the corresponding Bounded 2CL graph to true. 
The black player is given a piece that can move straight to the 
white den. This will take him so many moves, that if the 
corresponding Bounded 2CL graph can be set to true, by the time (s)he 
reaches it the white player has already won the game. \hfill$\Box$
\end{proof}

\section{Conclusions}
We reduced Acyclic Bounded NCL to \textsc{Klondike} and \textsc{Mahjong Solitaire} and (planar) Constraint Graph Satisfiability to \textsc{Nonogram},
proving them to be NP-complete. By using the acyclic property to our advantage,
we were able to keep the reductions elegant and easy to understand.
For games that require to return to an ``empty'' configuration (like
Klondike and Mahjong Solitaire) acyclicity is even technically essential.
We acknowledge the NCL framework to be well-suited for reductions for games, but
it is not without drawbacks. Often the primary gadgets are relatively
easy to construct, while the construction of victory gadgets is sometimes less trivial.
Finally, we reduced Planar Bounded 2CL to the game of Dou Shou Qi proving it to by PSPACE-hard.
The generic planarization of the NCL graphs and 2CL graphs is very useful for reductions
to games played on a 2-dimensional board. As an unbounded two-player
game Dou Shou Qi is expected to be EXPTIME-complete in the classification by~\citeaby{Hearn2009}.
It is an open problem to construct the Dou Shou Qi gadgets for the
special vertices, e.g., multiplayer AND, that build the relevant 2CL graph.

\section*{Acknowledgments}
The authors would like to thank the anonymous referees,
in particular for the suggested simplification regarding the Nonogram problem.

\bibliography{NCL}

\end{document}